\documentclass{llncs}

\usepackage[usenames,dvipsnames]{color}

\usepackage{latexsym}
\usepackage{amsxtra} 
\usepackage{amssymb}
\usepackage{amsmath}
\usepackage{pslatex}
\usepackage{epsfig}
\usepackage{wrapfig}
\usepackage{paralist}
\usepackage{graphics}
\usepackage{stmaryrd}
\usepackage{txfonts}
\usepackage{framed}

\usepackage{proof}

\usepackage{algorithm}
\usepackage{algorithmicx}
\usepackage[noend]{algpseudocode}

\newtheorem{fact}{Fact}

\newcommand{\rbr}{{\bf ]\!]}}
\newcommand{\lbr}{{\bf [\![}}
\newcommand{\sem}[1]{\lbr #1 \rbr}

\newcommand{\set}[1]{\left\{ #1 \right\}}
\newcommand{\tuple}[1]{\left\langle #1 \right\rangle}
\renewcommand{\vec}[1]{\mathbf #1}
\newcommand{\tup}[1]{\overline #1}
\newcommand{\ite}[3]{\mathsf{ite}(#1,#2,#3)}
\newcommand{\tupite}[2]{\mathsf{ite}(#1,#2)}

\newcommand{\len}[1]{{|{#1}|}}
\newcommand{\card}[1]{{|\!|{#1}|\!|}}

\newcommand{\proj}[2]{{#1}\!\!\downarrow_{{#2}}}

\renewcommand{\paragraph}[1]{\noindent\emph{#1}}

\newif\ifLongVersion\LongVersiontrue

\newcommand{\teq}{\approx}

\newcommand{\I}{\mathcal{I}}
\newcommand{\J}{\mathcal{J}}
\newcommand{\mods}{\mathbf{I}}
\newcommand{\modsof}[2]{\sem{#1}_{#2}}
\newcommand{\locs}{\mathsf{Loc}}
\newcommand{\data}{\mathsf{Data}}
\newcommand{\nil}{\mathsf{nil}}

\newcommand{\emp}{\mathsf{emp}}
\newcommand{\wand}{\multimap}
\newcommand{\seplog}{\mathsf{SL}}
\newcommand{\seplogl}[1]{\mathsf{SL}\parens{#1}}
\newcommand{\tterm}{\mathsf{t}}
\newcommand{\uterm}{\mathsf{u}}
\newcommand{\vterm}{\mathsf{v}}
\newcommand{\xterm}{\mathsf{x}}
\newcommand{\yterm}{\mathsf{y}}
\newcommand{\zterm}{\mathsf{z}}
\newcommand{\wterm}{\mathsf{w}}
\newcommand{\nextp}{\mathsf{next}}
\newcommand{\datap}{\mathsf{data}}
\newcommand{\lterm}{\ell}
\newcommand{\pto}{\mathsf{pt}}
\newcommand{\bools}{\mathsf{Bool}}
\newcommand{\sets}{\mathsf{Set}}
\newcommand{\lab}{\triangleleft}
\newcommand{\labd}[1]{#1\!\!\Downarrow}

\newcommand{\terms}[1]{\mathrm{Pt}(#1)}
\newcommand{\bounds}[2]{\mathrm{Bnd}(#1,#2)}
\newcommand{\dom}{\mathrm{dom}}
\newcommand{\size}[1]{\mathrm{Size}(#1)}

\newcommand{\parens}[1]{ \left( #1 \right) }

\newcommand{\purify}[1]{\lfloor #1 \rfloor}
\newcommand{\purifyrec}[1]{\purify{#1}^\ast}
\newcommand{\unpurify}[1]{\lceil #1 \rceil}
\newcommand{\quants}[1]{\mathrm{Q}(#1)}
\newcommand{\solvesl}[1]{\mathsf{solve_{SL(#1)}}}
\newcommand{\solveslt}[1]{\mathsf{solve_{#1}}}
\newcommand{\lblf}{\mathsf{lbl}}
\newcommand{\ltrue}{\top}
\newcommand{\lfalse}{\bot}
\newcommand{\ord}[2]{\prec_{#1,#2}}
\newcommand{\qvar}{x}

\newcommand{\sparagraph}[1]{\smallskip
\noindent 
{\bf #1}\ }

\begin{document}

\title{A Decision Procedure for Separation Logic in SMT}
\author{Andrew Reynolds\inst{1} \and Radu Iosif\inst{2} \and Cristina Serban\inst{2} \and Tim King\inst{3}}
\institute{The University of Iowa \and Universit\'e de Grenoble Alpes, CNRS, VERIMAG \and Google Inc.}
\maketitle

\begin{abstract}
This paper presents a complete decision procedure for the entire
quantifier-free fragment of Separation Logic ($\seplog$) interpreted
over heaplets with data elements ranging over a parametric
multi-sorted (possibly infinite) domain. The algorithm uses a
combination of theories and is used as a specialized solver inside a
DPLL($T$) architecture. A prototype was implemented within the CVC4
SMT solver. Preliminary evaluation suggests the possibility of using
this procedure as a building block of a more elaborate theorem prover
for $\seplog$ with inductive predicates, or as back-end of a bounded model
checker for programs with low-level pointer and data manipulations.
\end{abstract}

\section{Introduction}

Separation Logic ($\seplog$) \cite{Reynolds02} is a logical framework
for describing dynamically allocated mutable data structures generated
by programs that use pointers and low-level memory allocation
primitives. The logics in this framework are used by a number of
academic (\textsc{Space Invader} \cite{SpaceInvader}, \textsc{Sleek}
\cite{Sleek}), and industrial (\textsc{Infer} \cite{Infer}) tools for
program verification. The main reason for chosing to work within the
$\seplog$ framework is its ability to provide compositional proofs of
programs, based on the principle of {\em local reasoning}: analyzing
different parts of the program (e.g.\ functions, threads), that work
on \emph{disjoint parts of the global heap}, and combining the
analysis results a-posteriori. 

The main ingredients of $\seplog$ are \begin{inparaenum}[(i)] 
\item the \emph{separating conjunction} $\phi * \psi$, which
  asserts that $\phi$ and $\psi$ hold for separate portions of the
  memory (heap), 
\item the \emph{magic wand} $\varphi \wand \psi$, which asserts that
  any extension of the heap by a disjoint heap that satisfies $\phi$
  must satisfy $\psi$, and
\item the \emph{frame rule}, that exploits separation to provide
  modular reasoning about programs.
\end{inparaenum}
Consider, for instance, a memory configuration (heap), in which two
cells are allocated, and pointed to by the program variables
$\mathsf{x}$ and $\mathsf{y}$, respectively, where the $\mathsf{x}$
cell has an outgoing selector field to the $\mathsf{y}$ cell, and
viceversa. The heap can be split into two disjoint parts, each
containing exactly one cell, and described by an atomic proposition
$\mathsf{x} \mapsto \mathsf{y}$ and $\mathsf{y} \mapsto \mathsf{x}$,
respectively. Then the entire heap is described by the formula
$\mathsf{x} \mapsto \mathsf{y} * \mathsf{y} \mapsto \mathsf{x}$, which
reads \emph{$\mathsf{x}$ points to $\mathsf{y}$ and separately
  $\mathsf{y}$ points to $\mathsf{x}$}.

The expressive power of $\seplog$ comes with an inherent difficulty of
automatically reasoning about the satisfiability of its formulae, as
required by push-button program analysis tools. Indeed, $\seplog$
becomes undecidable in the presence of first-order quantification,
even when the fragment uses only points-to predicates, without the
separating conjunction or the magic
wand~\cite{CalcagnoYangOHearn01}. Moreover, the quantifier-free
fragment with no data constraints, using only points-to predicates
$\mathsf{x} \mapsto (\mathsf{y},\mathsf{z})$, where $\mathsf{x},
\mathsf{y}$ and $\mathsf{z}$ are interpreted as memory addresses, is
PSPACE-complete, due to the implicit quantification over memory
partitions, induced by the semantics of the separation logic
connectives, which can, moreover, be arbitrarily
nested~\cite{CalcagnoYangOHearn01}.

This paper presents a decision procedure for quantifier-free $\seplog$
which is entirely parameterized by a base theory $T$ of heap locations
and data, i.e.\ the sorts of memory addresses and their contents can
be chosen from a large variety of available theories handled by
Satisfiability Modulo Theories (SMT) solvers, such as linear integer
(real) arithmetic, strings, sets, uninterpreted functions, etc. Given
a base theory $T$, we call $\seplog(T)$ the set of separation logic
formulae built on top of $T$, by considering points-to predicates and
the separation logic connectives. 

\sparagraph{Contributions} First, we show that the satisfiability
problem for the quantifier-free fragment of $\seplog(T)$ is
\textsc{PSPACE}-complete, provided that the satisfiability of the
quantifier-free fragment of the base theory $T$ is in
\textsc{PSPACE}. Our method is based on a semantics-preserving
translation of $\seplog(T)$ into second-order $T$ formulae with
quantifiers over a domain of sets and uninterpreted functions, whose
cardinality is polynomially bound by the size of the input
formula. For the fragment of $T$ formulae produced by the translation
from $\seplog(T)$, we developped a lazy quantifier instantiation
method, based on counterexample-driven refinement. We show that the
quantifier instantiation algorithm is sound, complete and terminates
on the fragment under consideration. We present our algorithm for the
satisfiability of quantifier-free $\seplog(T)$ logics as a component
of a DPLL($T$) architecture
\cite{GanzingerHagenNieuwenhuisOliverasTinelli04}, which is widely
used by modern SMT solvers. We have implemented a prototype solver as
a branch of the CVC4 SMT solver ~\cite{CVC4-CAV-11} and carried out
experiments that handle non-trivial examples quite
effectively. Applications of our procedure
include:~\begin{compactenum}
\item Integration within theorem provers for $\seplog$ with inductive
  predicates. Most inductive provers for $\seplog$ use a high-level
  proof search strategy relying on a separate decision procedure for
  entailments in the non-inductive fragment, used to simplify the
  proof obligations, by discharging the non-inductive parts of both
  left- and right-hand sides, and attain an inductive hypothesis
  \cite{Sleek,cyclist}. Due to the hard problem of proving entailments
  in the non-inductive fragment of $\seplog$, these predicates use
  very simple non-inductive formulae (a list of points-to propositions
  connected with separating conjunction), for which entailments are
  proved by syntactic substitutions and matching. Our work aims at
  extending the language of inductive $\seplog$ solvers, by
  outsourcing entailments in a generic non-inductive fragment to a
  specialized procedure. To this end, we conducted experiments on
  several entailments corresponding to finite unfoldings of inductive
  predicates used in practice (Section \ref{sec:eval}).
\item Use as back-end of a bounded model checker for programs with
  pointer and data manipulations, based on a complete weakest
  precondition calculus that involves the magic wand
  connective~\cite{IshtiaqOHearn01}. To corroborate this hypothesis,
  we tested our procedure on verification conditions automatically
  generated by applying the weakest precondition calculus described
  in~\cite{IshtiaqOHearn01} to several program fragments (Section
  \ref{sec:eval}).
\end{compactenum} 

\sparagraph{Related Work} The study of the algorithmic properties of
Separation Logic \cite{Reynolds02} has produced an extensive body of
literature over time. We need to distinguish between $\seplog$ with
inductive predicates and restrictive non-inductive fragments, and
$\seplog$ without inductive predicates, which is the focus of this
paper.

Regarding $\seplog$ with fixed inductive predicates, Perez and
Rybalchenko \cite{NavarroPerez2011} define a theorem proving framework
relying on a combination of $\seplog$ inference rules dealing with
singly-linked lists only, and a superposition calculus dealing with
equalities and aliasing between variables. Concerning $\seplog$ with
generic user-provided inductive predicates, the theorem prover
\textsc{Sleek} \cite{Sleek} implements a semi-algorithmic entailment
check, based on unfoldings and unifications. Along this line of work,
the theorem prover \textsc{Cyclist} \cite{cyclist} builds entailment
proofs using a sequent calculus. More recently, the tool
\textsc{Slide} \cite{slide} reduces the entailment between inductive
predicates to an inclusion between tree automata. The great majority
of these inductive provers focus on applying induction strategies
efficiently, and consider a very simple fragment of non-inductive
$\seplog$ formulae, typically conjunctions of equalities and
disequalities between location variables and separated points-to
predicates, without negations or the magic wand. On a more general
note, the tool \textsc{Spen} \cite{spen} considers also arithmetic
constraints between the data elements in the memory cells, but fixes
the shape of the user-defined predicates.

The idea of applying SMT techniques to decide satisfiability of
$\seplog$ formulae is not new. In their work, Piskac, Wies and
Zufferey translate from $\seplog$ with singly-linked list segments
\cite{Piskac2013} and trees \cite{Piskac2014}, respectively, into
first-order logics (\textsc{Grass} and \textsc{Grit}) that are
decidable in $\textsc{NP}$. The fragment handled in this paper is
incomparable to the logics \textsc{Grass} \cite{Piskac2013} and
\textsc{Grit} \cite{Piskac2014}. On one hand, we do not consider
predicates defining recursive data structures, such as singly-linked
lists. On the other hand, we deal with the entire quantifier-free
fragment of $\seplog$, including arbitrary nesting of the magic wand,
separating conjunction and classical boolean connectives. As a result,
the decision problem we consider is \textsc{PSPACE}-complete, due to
the possibility of arbitrary nesting of the boolean and $\seplog$
connectives. To the best of our knowledge, our implementation is also
the first to enable theory combination involving $\seplog$, in a
fine-grained fashion, directly within the DPLL($T$) loop. Indeed, we
do not translate $\seplog$ into classical multi-sorted logic upfront,
but rather use lazy evaluation and counterexample-driven quantifier
instantiation, that can solve non-trivial instances of the
satisfiability problem efficiently.

The first theoretical results on decidability and complexity of
$\seplog$ without inductive predicates are given by Calcagno, Yang and
O'Hearn \cite{CalcagnoYangOHearn01}. They show that the
quantifier-free fragment of $\seplog$ without data constraints is
PSPACE-complete by an argument that enumerates a finite (yet large)
set of heap models. Their argument shows also the difficulty of the
problem, however it cannot be directly turned into an effective
decision procedure, because of the ineffectiveness of model
enumeration. A more elaborate tableau-based decision procedure is
described by M\'ery and Galmiche \cite{MeryGalmiche07}. This procedure
generates verification conditions on-demand, but here no data
constraints are considered, either. However, combined, the results of
\cite{CalcagnoYangOHearn01} and \cite{MeryGalmiche07} have inspired
our decision procedure, that is parameterized by a decidable
quantifier-free fragment of a base theory, and can be used in
combination with any available SMT theory.

Our procedure relies on a decision procedure for quantifier-free
parametric theory of sets and on-demand techniques for quantifier
instantiation.  Decision procedures for the theory of sets in SMT are
given in~\cite{suter2011sets,BansalPhd}.  Techniques for model-driven
quantifier instantiation were introduced in the context of SMT
in~\cite{GeDeM-CAV-09}, and have been developed recently
in~\cite{ReynoldsDKBT15Cav,bjornerplaying}.

\section{Preliminaries}

We consider formulae in multi-sorted first-order logic, over a
\emph{signature} $\Sigma$ consisting of a countable set of sort
symbols and a set of function symbols.  We assume that signatures
always include a boolean sort $\bools$ with constants $\ltrue$ and
$\lfalse$ denoting true and false respectively, and that each sort
$\sigma$ is implicitly equipped with an equality predicate $\teq$ over
$\sigma \times \sigma$.  Moreover, we may assume without loss of
generality that equality is the only predicate belonging to $\Sigma$,
since we can model other predicate symbols as function symbols with
return sort $\bools$\footnote{For brevity, we may write $p( \vec t )$
  as shorthand for $p( \vec t ) \teq \ltrue$, where $p$ is a function
  into $\bools$.}.

We consider a set $\vec x$ of first-order variables, with associated
sorts, and denote by $\varphi(\vec x)$ the fact that the free
variables of the formula $\varphi$ belong to $\vec x$. Given a
signature $\Sigma$, well-sorted terms, atoms, literals, and formulae
are defined as usual, and referred to respectively as
\emph{$\Sigma$-terms}. We denote by $\phi[\varphi]$ the fact that
$\varphi$ is a subformula (subterm) of $\phi$ and by
$\phi[\psi/\varphi]$ the result of replacing $\varphi$ with $\psi$ in
$\phi$. We write $\forall x. \varphi$ to denote universal
quantification over variable $x$, where $x$ occurs as a \emph{free
  variable} in $\varphi$. If $\vec x = \tuple{x_1,\ldots,x_n}$ is a
tuple of variables, we write $\forall. \vec x\, \varphi$ as an
abbreviation of $\forall. x_1 \cdots \forall x_n\, \varphi$.  We say
that a $\Sigma$-term is \emph{ground} if it contains no free
variables.  We assume $\Sigma$ contains an if-then-else operator
$\ite{b}{\tterm}{\uterm}$, of sort $\bools \times \sigma \times \sigma
\to \sigma$, for each sort $\sigma$, that evaluates to $\tterm$ if
$b=\top$ and to $\uterm$ otherwise.

A \emph{$\Sigma$-interpretation $\I$} maps:\begin{inparaenum}[(i)]
\item each set sort symbol $\sigma \in \Sigma$ to a non-empty set
  $\sigma^\I$, the \emph{domain} of $\sigma$ in $\I$, 
\item each function symbol $f \in \Sigma$ of sort $\sigma_1 \times
  \ldots \times \sigma_n \rightarrow \sigma$ to a total function
  $f^\I$ of sort $\sigma^\I_1 \times \ldots \times \sigma^\I_n
  \rightarrow \sigma^\I$ if $n > 0$, and to an element of $\sigma^\I$
  if $n = 0$, and
\item each variable $x\in\vec x$ to an element of $\sigma_x^\I$, where
  $\sigma_x$ is the sort symbol associated with $x$.
\end{inparaenum}
We denote by $\tterm^\I$ the interpretation of a term $\tterm$ induced
by the mapping $\I$. The satisfiability relation between
$\Sigma$-interpretations and $\Sigma$-formulae, written $\I \models
\varphi$, is defined inductively, as usual. We say that $\I$ is
\emph{a model of $\varphi$} if $\I\models\varphi$.

A \emph{first-order theory} is a pair $T = (\Sigma, \mods)$ where
$\Sigma$ is a signature and $\mods$ is a non-empty set of
$\Sigma$-interpretations, the \emph{models} of $T$. For a formula
$\varphi$, we denote by $\modsof{\varphi}{T} = \set{\I \in \mods \mid
  \I \models \varphi}$ its set of $T$-models. A $\Sigma$-formula
$\varphi$ is \emph{$T$-satisfiable} if $\modsof{\varphi}{T} \neq
\emptyset$, and \emph{$T$-unsatisfiable} otherwise. A $\Sigma$-formula
$\varphi$ is \emph{$T$-valid} if $\modsof{\varphi}{T} = \mods$,
i.e.\ if $\neg\varphi$ is $T$-unsatisfiable. A formula $\varphi$
\emph{$T$-entails} a $\Sigma$-formula $\phi$, written $\varphi
\models_T \phi$, if every model of $T$ that satisfies $\varphi$ also
satisfies $\phi$. The formulae $\varphi$ and $\psi$ are
\emph{$T$-equivalent} if $\varphi \models_T \phi$ and $\phi \models_T
\varphi$, and \emph{equisatisfiable} (\emph{in $T$}) if $\psi$ is
$T$-satisfiable if and only if $\varphi$ is $T$-satisfiable.
Furthermore, formulas $\varphi$ and $\psi$ are \emph{equivalent (up to
  $\vec k$)} if they are satisfied by the same set of models (when
restricted to the interpretation of variables $\vec k$).  The
$T$-satisfiability problem asks, given a $\Sigma$-formula $\varphi$,
whether $\sem{\varphi}_T \neq \emptyset$, i.e.\ whether $\varphi$ has
a $T$-model.

\subsection{Separation Logic} 

In the remainder of the paper we fix a theory $T=(\Sigma,\mods)$, such
that the $T$-satisfiability for the language of quantifier-free
boolean combinations of equalities and disequalties between
$\Sigma$-terms is decidable. We fix two sorts $\locs$ and $\data$ from
$\Sigma$, with no restriction other than the fact that $\locs$ is
always interpreted as a countably infinite set. We refer to
\emph{Separation Logic} for $T$, written $\seplogl{T}$, as the set of
formulae generated by the syntax:
\[\begin{array}{lcl}
\phi & := & \tterm \teq \uterm \mid \tterm \mapsto \uterm \mid \emp 
\mid \phi_1 * \phi_2 \mid \phi_1 \wand \phi_2 \mid \phi_1 \wedge \phi_2 \mid \neg\phi_1
\end{array}\]
where $\tterm$ and $\uterm$ are well-sorted $\Sigma$-terms and that
for any atomic proposition $\tterm \mapsto \uterm$, $\tterm$ is of
sort $\locs$ and $\uterm$ is of sort $\data$. Also, we consider that
$\Sigma$ has a constant $\nil$ of sort $\locs$, with the intended
meaning that $\tterm \mapsto \uterm$ never holds when
$\tterm\teq\nil$. In the following, we write $\phi\vee\psi$ for
$\neg(\neg\phi\wedge\neg\psi)$ and $\phi\Rightarrow\psi$ for
$\neg\phi\vee\psi$.

Given an interpretation $\I$, a \emph{heap} is a finite partial
mapping $h : \locs^\I \rightharpoonup_{\mathrm{fin}} \data^\I$. For a
heap $h$, we denote by $\dom(h)$ its domain. For two heaps
$h_1$ and $h_2$, we write $h_1 \# h_2$ for $\dom(h_1) \cap
\dom(h_2) = \emptyset$ and $h = h_1 \uplus h_2$ for $h_1 \#
h_2$ and $h = h_1 \cup h_2$. For an interpretation $\I$, a heap $h :
\locs^\I \rightharpoonup_{\mathrm{fin}} \data^\I$ and a $\seplog(T)$
formula $\phi$, we define the satisfaction relation $\I,h
\models_{\seplog} \phi$ inductively, as follows:
\[\begin{array}{lcl}
\I,h \models_{\seplog} \emp & \iff & h = \emptyset \\
\I,h \models_{\seplog} \tterm \mapsto \uterm & \iff & h = \{(\tterm^\I,\uterm^\I)\} 
\text{ and } \tterm^\I\not\teq\nil^\I \\
\I,h \models_{\seplog} \phi_1 * \phi_2 & \iff & \exists h_1,h_2 \text{ . } h=h_1\uplus h_2
\text{ and } \I,h_i \models_{\seplog} \phi_i, \text{ for all } i = 1,2 \\
\I,h \models_{\seplog} \phi_1 \wand \phi_2 & \iff & \forall h' \text{ if } h'\#h \text{ and } 
\I,h' \models_{\seplog} \phi_1 \text{ then } \I,h'\uplus h \models_{\seplog} \phi_2
\end{array}\]
The satisfaction relation for the equality atoms $\tterm \teq \uterm$
and the Boolean connectives $\wedge$, $\neg$ are the classical ones from
first-order logic.
The \emph{$(\seplog,T)$-satisfiability} problem asks whether there is a
$T$-model $\I$ such that $(\I,h) \models_{\seplog}$ for some heap $h$.

In this paper we tackle the $(\seplog,T)$-satisfiability problem,
under the assumption that the quantifier-free data theory
$T=(\Sigma,\mods)$ has a decidable satisfiability problem for
constraints involving $\Sigma$-terms. It has been proved
\cite{CalcagnoYangOHearn01} that the satisfiability problem is
PSPACE-complete for the fragment of separation logic in which $\data =
\locs \times \locs$. Here we generalize this result to any combination
of theories whose satisfiability, for the quantifier-free fragment, is
in \textsc{PSPACE}. This is, in general, the case of most SMT
theories, which are typically in NP, such as the linear arithmetic of
integers and reals, possibly with sets and uninterpreted functions,
etc.

\section{Reducing $\seplogl{T}$ to Multisorted Second-Order Logic} 

It is well-known \cite{Reynolds02} that separation logic cannot be
formalized as a classical (unsorted) first-order theory, for instance,
due to the behavior of the $*$ connective, that does not comply with
the standard rules of contraction $\phi \Rightarrow \phi * \phi$ and
weakening $\phi * \varphi \Rightarrow \phi$\footnote{Take for instance
  $\phi$ as $x \mapsto 1$ and $\varphi$ as $y \mapsto 2$.}. The basic
reason is that $\phi * \varphi$ requires that $\phi$ and $\varphi$
hold on \emph{disjoint} heaps. Analogously, $\phi \wand \varphi$ holds
on a heap whose extensions, by disjoint heaps satisfying $\phi$, must
satisfy $\varphi$. In the following, we leverage from the expressivity
of multi-sorted first-order theories and translate $\seplog(T)$
formulae into quantified formulae in the language of $T$, assuming
that $T$ subsumes a theory of sets and uninterpreted functions.

The integration of separation logic within the DPLL(T) framework
\cite{GanzingerHagenNieuwenhuisOliverasTinelli04} requires the input
logic to be presented as a multi-sorted logic. To this end, we assume,
without loss of generality, the existence of a fixed theory
$T=(\Sigma,\mods)$ that subsumes a theory of sets $\sets(\sigma)$
\cite{BansalPhd}, for any sort $\sigma$ of set elements, whose
functions are the union $\cup$, intersection $\cap$ of sort
$\sets(\sigma) \times \sets(\sigma) \rightarrow \sets(\sigma)$,
singleton $\set{.}$ of sort $\sigma \rightarrow \sets(\sigma)$ and
emptyset $\emptyset$ of sort $\sets(\sigma)$. We write $\lterm
\subseteq \lterm'$ as a shorthand for $\lterm \cup \lterm' \teq
\lterm'$ and $\tterm \in \lterm$ for $\set{\tterm} \subseteq \lterm$,
for any terms $\lterm$ and $\lterm'$ of sort $\sets(\sigma)$ and
$\tterm$ of sort $\sigma$. The interpretation of the functions in the
set theory is the classical (boolean) one.

Also, we assume without loss of generality the existence of infinitely
many function symbols $\pto,\pto',\ldots \in \Sigma$ of sort $\locs
\mapsto \data$, where $\locs$ and $\data$ are two fixed sorts of $T$,
such that for any interpretation $\I\in\mods$, $\locs^\I$ is an
infinite countable set\footnote{The generalization of $\seplog(T)$ to
  finite interpretations of $\locs$ is considered as future work.}.
We do not assume any particular interpretation of these symbols in the
following and treat them as uninterpreted function symbols.


The main idea is to express the atoms and connectives of separation
logic in multi-sorted second-order logic by means of a transformation,
called \emph{labeling}, which introduces \begin{inparaenum}[(i)]
\item constraints over variables of sort $\sets(\locs)$ and 
\item terms over uninterpreted \emph{points-to} functions of sort
  $\locs\rightarrow\data$.
\end{inparaenum} 
We describe the labeling transformation using judgements of the form
$\phi \lab [\tup\lterm,\tup\pto]$, where $\phi$ is a $\seplog(T)$
formula, $\tup\lterm = \tuple{\lterm_1,\ldots, \lterm_n}$ is a tuple
of variables of sort $\sets(\locs)$ and $\tup\pto =
\tuple{\pto_1,\ldots,\pto_n}$ is a tuple of uninterpreted function
symbols occurring under the scopes of universal quantifiers. To ease
the notation, we write $\lterm$ and $\pto$ instead of the singleton
tuples $\tuple{\lterm}$ and $\tuple{\pto}$. In the following, we also
write $\bigcup\tup\lterm$ for $\lterm_1 \cup \ldots \cup \lterm_n$,
$\lterm' \cap \tup\lterm$ for $\tuple{\lterm' \cap \lterm_1, \ldots,
  \lterm' \cap \lterm_n}$, $\lterm' \cdot \tup\lterm$ for
$\tuple{\lterm', \lterm_1, \ldots, \lterm_n}$ and
$\tupite{\tterm\in\tup\lterm}{\tup\pto(\tterm)=\uterm}$ for
$\ite{\tterm\in\lterm_1}{\pto_1(\tterm)=\uterm}{\ite{\tterm\in\lterm_2}{\pto_2(\tterm)=\uterm}{
    \ldots,\ite{\tterm\in\lterm_n}{\pto_n(\tterm)=\uterm}{\top}\ldots}}$.

Intuitively, a labeled formula $\phi \lab [\tup\lterm,\tup\pto]$ says
that it is possible to build, from any of its satisfying
interpretations $\I$, a heap $h$ such that $\I,h \models_{\seplog}
\phi$, where $\dom(h)=\lterm_1^\I \cup \ldots \cup
\lterm_n^\I$ and $h = \proj{\pto_1^\I}{\lterm_1^\I} \cup \ldots \cup
\proj{\pto_n^\I}{\lterm_n^\I}$\footnote{We denote by $\proj{F}{D}$ the
  restriction of the function $F$ to the domain
  $D\subseteq\dom(F)$.}. More precisely, a variable $\lterm_i$
defines a slice of the domain of the (global) heap, whereas the
restriction of $\pto_i$ to (the interpretation of) $\lterm_i$
describes the heap relation on that slice. Observe that each
interpretation of $\tup\lterm$ and $\tup\pto$, such that
$\lterm^\I_i \cap \lterm^\I_j = \emptyset$, for all $i \neq j$,
defines a unique heap.

First, we translate an input $\seplog(T)$ formula $\phi$ into a
labeled second-order formula, with quantifiers over sets and
uninterpreted functions, defined by the rewriting rules in Figure
\ref{fig:labeling-rules}. A labeling step $\phi[\varphi]
\Longrightarrow \phi[\psi/\varphi]$ applies if $\varphi$ and $\psi$
match the antecedent and consequent of one of the rules in Figure
\ref{fig:labeling-rules}, respectively. It is not hard to show that
this rewriting system is confluent, and we denote by $\labd{\phi}$ the
normal form of $\phi$ with respect to the application of labeling
steps. 

\begin{figure}[htb]
\vspace*{-\baselineskip}
\[\begin{array}{ccl}
\infer[]{\neg\forall \lterm_1 \forall \lterm_2 ~.~ \neg(\lterm_1 \cap \lterm_2 \teq \emptyset \wedge \lterm_1 \cup \lterm_2 \teq \bigcup\lterm
\wedge \phi\lab[\lterm_1\cap\tup\lterm,\tup\pto] \wedge \psi\lab[\lterm_2\cap\tup\lterm,\tup\pto])}{(\phi * \psi) \lab [\tup\lterm,\tup\pto]} 
& \hspace*{0.8cm} &
\infer[]{\phi \lab [\tup\lterm,\tup\pto] \wedge \psi \lab [\tup\lterm,\tup\pto]}{(\phi \wedge \psi) \lab [\tup\lterm,\tup\pto]} 
\\\\
\infer[]{
\forall \lterm' \forall \pto' ~.~ (\lterm'\cap(\bigcup\tup\lterm) \teq \emptyset \wedge \phi \lab [\lterm',\pto']) 
\Rightarrow \psi \lab [\lterm'\cdot\tup\lterm,\pto'\cdot\tup\pto]}{
(\phi \wand \psi) \lab [\tup\lterm,\tup\pto]}
& \hspace*{0.8cm} &
\infer[]{\neg(\phi \lab [\tup\lterm,\tup\pto])}{(\neg\phi) \lab [\tup\lterm,\tup\pto]}
\\\\
\infer[]{\bigcup\tup\lterm \teq \set{\tterm} \wedge 
\tupite{\tterm\in\tup\lterm}{\tup\pto(\tterm) \teq \uterm} \wedge \tterm\not\teq\nil}{\tterm \mapsto \uterm \lab [\tup\lterm,\tup\pto]}
\hspace*{1cm} 
\infer[]{\bigcup\tup\lterm \teq \emptyset}{\emp \lab [\tup\lterm,\tup\pto]}
& \hspace*{0.8cm} &
\infer[\text{if } \varphi \text{ is a $\Sigma$-formula}]{\varphi}{\varphi \lab [\tup\lterm,\tup\pto]}
\end{array}\] 
\caption{Labeling Rules}
\label{fig:labeling-rules}
\vspace*{-\baselineskip}
\end{figure}

\begin{example}\label{ex:labeling2}
Consider the $\seplog(T)$ formula $(x \mapsto a \wand y \mapsto b)
\wedge \emp$. The reduction to second-order logic is given below:
\[\begin{array}{lcr}
((x \mapsto a \wand y \mapsto b) \wedge \emp) \lab [\lterm,\pto] & \hspace*{1mm}& \stackrel{\wedge}{\Longrightarrow} \\[-1mm] 
(x \mapsto a \wand y \mapsto b) \lab [\lterm,\pto] \wedge \emp \lab [\lterm,\pto] & \hspace*{1mm}& \stackrel{\emp}{\Longrightarrow} \\[-1mm]
(x \mapsto a \wand y \mapsto b) \lab [\lterm,\pto] \wedge \lterm\teq \emptyset & \hspace*{1mm}& \stackrel{\wand}{\Longrightarrow} \\[-1mm]
\lterm\teq \emptyset \wedge \forall \lterm' \forall \pto' ~.~ \lterm'\cap\lterm\teq \emptyset \wedge (x \mapsto a \lab [\lterm',\pto']) \Rightarrow 
y \mapsto b \lab [\tuple{\lterm',\lterm},\tuple{\pto',\pto}]  & \hspace*{1mm}& \stackrel{\mapsto}{\Longrightarrow} \\[-1mm]
\\[-2mm]
\begin{array}{l}
\lterm\teq \emptyset \wedge \forall \lterm' \forall \pto' ~.~ \lterm'\cap\lterm\teq \emptyset \wedge 
\lterm'\teq \set{x} \wedge \ite{x\in\lterm'}{\pto'(x)\teq a}{\top} \wedge x\not\teq\nil \Rightarrow \\[-1mm]
\lterm'\cup\lterm\teq \set{y} \wedge \ite{y\in\lterm'}{\pto'(y)\teq b}{\ite{y\in\lterm}{\pto(y)\teq b}{\top}} \wedge y\not\teq\nil 
\enspace\enspace\blacksquare
\end{array} 
\end{array}\]
\end{example}

The following lemma reduces the $(\seplog,T)$-satisfiability problem
to the satisfiability of a quantified fragment of the multi-sorted
second-order theory $T$, that contains sets and uninterpreted
functions. For an interpretation $\I$, a variable $x$ and a value $s \in
\sigma_{\!\!x}^\I$, we denote by $\I[x\leftarrow s]$ the extension of
$\I$ which maps $x$ into $s$ and behaves like $\I$ for all other
symbols. We extend this notation to tuples
$\tup{x}=\tuple{x_1,\ldots,x_n}$ and $\tup{s}=\tuple{s_1,\ldots,s_n}$
and write $\I[\tup{x}\leftarrow\tup{s}]$ for $\I[x_1\leftarrow s_1]
\ldots [x_n \leftarrow s_n]$. For a tuple of heaps
$\tup{h}=\tuple{h_1,\ldots,h_n}$ we write $\dom(\tup{h})$ for
$\tuple{\dom(h_1), \ldots, \dom(h_n)}$. 

\begin{lemma}\label{lemma:labeling}
  Given a $\seplog(T)$ formula $\varphi$ and tuples $\tup\lterm =
  \tuple{\lterm_1,\ldots,\lterm_n}$ and $\tup\pto =
  \tuple{\pto_1,\ldots,\pto_n}$ for $n>0$, for any interpretation $\I$
  of $T$ and any heap $h$: \(\I,h \models_{\seplog} \varphi\) if and
  only if \begin{compactenum}
  \item\label{it1:labeling} for all heaps $\tup{h}=\tuple{h_1,\ldots,h_n}$ such that
  $h=h_1 \uplus \ldots \uplus h_n$, 
  \item\label{it2:labeling} for all heaps $\tup{h'}=\tuple{h'_1,\ldots,h'_n}$ such that
    $h_1\subseteq h'_1, \ldots, h_n\subseteq h'_n$,
  \end{compactenum}
  we have
  \(\I[\tup\lterm \leftarrow
    \dom(\tup{h})][\tup\pto \leftarrow \tup{h'}] \models_T
  \labd{\varphi \lab [\tup\lterm,\tup\pto]}\enspace.\)
\end{lemma}
Although, in principle, satisfiability is undecidable in the presence
of quantifiers and uninterpreted functions, the result of the next
section strenghtens this reduction, by adapting the labeling rules for
$*$ and $\wand$ (Figure \ref{fig:labeling-rules}) to use bounded
quantification over finite (set) domains.

\section{A Reduction of $\seplogl{T}$ to Quantifiers Over Bounded Sets} 
\label{sec:boundq}

In the previous section, we have reduced any instance of the
$(\seplog,T)$-satisfiability problem to an instance of the
$T$-satisfiability problem in the second-order multi-sorted theory $T$
which subsumes the theory $\sets(\locs)$ and contains several
quantified uninterpreted function symbols of sort $\locs \mapsto
\data$. A crucial point in the translation is that the only
quantifiers occurring in $T$ are of the forms $\forall \lterm$ and
$\forall \pto$, where $\lterm$ is a variable of sort $\sets(\locs)$
and $\pto$ is a function symbol of sort $\locs \mapsto
\data$. Leveraging from a small model property for $\seplog$ over the
data domain $\data = \locs \times \locs$ \cite{CalcagnoYangOHearn01},
we show that it is sufficient to consider only the case when the
quantified variables range over a bounded domain of sets. In
principle, this allows us to eliminate the universal quantifiers by
replacing them with finite conjunctions and obtain a decidability
result based on the fact that the quantifier-free theory $T$ with sets
and uninterpreted functions is decidable. Since the cost of a-priori
quantifier elimination is, in general, prohibitive, in the next
section we develop an efficient lazy quantifier instantiation
procedure, based on counterexample-driven refinement.

For reasons of self-containment, we quote the following lemma
\cite{YangPhD} and stress the fact that its proof is oblivious of the
assumption $\data=\locs\times\locs$ on the range of heaps. Given a
formula $\phi$ in the language $\seplog(T)$, we first define the
following measure:
\[\begin{array}{ccccccc}
\len{\phi*\psi} = \len{\phi}+\len{\psi} & \hspace*{1cm} & 
\len{\phi \wand \psi} = \len{\psi} & \hspace*{1cm} & 
\len{\phi\wedge\psi} = \max(\len{\phi},\len{\psi}) & \hspace*{1cm} & 
\len{\neg\phi} = \len{\phi} \\
\len{\tterm \mapsto \uterm} = 1 & \hspace*{1cm} & 
\len{\emp} = 1 & \hspace*{1cm} & 
\len{\phi} = 0 \text{ if $\phi$ is a $\Sigma$-formula}
\end{array}\]
Intuitively, $\len{\phi}$ gives the maximum number of \emph{invisible}
locations in the domain of a heap $h$, that are not in the range of
$\I$ and which can be distinguished by $\phi$. For instance, if $\I,h
\models_{\seplog} (\neg \emp) * (\neg \emp)$ and the domain of $h$
contains more than two locations, then it is possible to restrict
$\dom(h)$ to two locations only, and obtain $h'$ such that
$\card{\dom(h')}=\len{(\neg \emp) * (\neg \emp)}=2$ and $\I,h'
\models_{\seplog} (\neg \emp) * (\neg \emp)$.

Let $\terms{\phi}$ be the set of terms (of sort $\locs \cup \data$)
that occur on the left- or right-hand side of a points-to atomic
proposition in $\phi$. Formally, we have $\terms{\tterm \mapsto
  \uterm}=\set{\tterm,\uterm}$,
$\terms{\phi*\psi}=\terms{\phi\wand\psi}=\terms{\phi}\cup\terms{\psi}$,
$\terms{\neg\phi}=\terms{\phi}$ and
$\terms{\emp}=\terms{\phi}=\emptyset$, for a $\Sigma$-formula
$\phi$. The small model property is given by the next lemma: 

\begin{lemma}\cite[Proposition 96]{YangPhD}\label{lemma:decidability}
Given a formula $\phi \in \seplog(T)$, for any interpretation $\I$ of
$T$, let \(L \subseteq \locs^\I \setminus \terms{\phi}^\I\) be a set
of locations, such that $\card{L}=\len{\phi}$ and $v \in \data^\I
\setminus \terms{\phi}^I$. Then, for any heap $h$, we have $\I,h
\models_\seplog \phi$ iff $\I,h' \models_\seplog \phi$, for any heap
$h'$ such that: \begin{compactitem}[-]
\item $\dom(h') \subseteq L \cup \terms{\phi}^\I$, 
\item for all $\ell \in \dom(h')$, $h'(\ell) \in \terms{\phi}^\I \cup
  \set{v}$
\end{compactitem}
\end{lemma}


Based on the fact that the proof of Lemma \ref{lemma:decidability}
\cite{YangPhD} does not involve reasoning about data values, other
than equality checking, we refine our reduction from the previous
section, by bounding the quantifiers to finite sets of constants of
known size. To this end, we assume the existence of a total order on
the (countable) set of constants in $\Sigma$ of sort $\locs$, disjoint
from any $\Sigma$-terms that occur in a given formula $\phi$, and
define $\bounds{\phi}{C} = \{c_{m+1}, \ldots, c_{m+\len{\phi}}\}$,
where $m=\max\{i \mid c_i \in C\}$, and $m=0$ if
$C=\emptyset$. Clearly, we have
$\terms{\phi}\cap\bounds{\phi}{C}=\emptyset$ and also $C \cap
\bounds{\phi}{C} = \emptyset$, for any $C$ and any $\phi$.

We now consider labeling judgements of the form $\varphi \lab
[\tup\lterm,\tup\pto,C]$, where $C$ is a finite set of constants of
sort $\locs$, and modify all the rules in Figure
\ref{fig:labeling-rules}, besides the ones with premises $(\phi *
\psi) \lab [\tup\lterm,\tup\pto]$ and $(\phi \wand \psi) \lab
    [\tup\lterm,\tup\pto]$, by replacing any judgement $\varphi \lab
    [\tup\lterm,\tup\pto]$ with $\varphi \lab
    [\tup\lterm,\tup\pto,C]$. The two rules in Figure
    \ref{fig:bounded-labeling} are the bounded-quantifier equivalents
    of the $(\phi * \psi) \lab [\tup\lterm,\tup\pto]$ and $(\phi \wand
    \psi) \lab [\tup\lterm,\tup\pto]$ rules in Figure
    \ref{fig:labeling-rules}. As usual, we denote by $\labd{(\varphi
      \lab [\tup\lterm,\tup\pto,C])}$ the formula obtained by
    exhaustively applying the new labeling rules to $\varphi \lab
    [\tup\lterm,\tup\pto,C]$.

Observe that the result of the labeling process is a formula in which
all quantifiers are of the form $\forall \lterm_1 \ldots \forall
\lterm_n \forall \pto_1 \ldots \forall \pto_n. \bigwedge_{i=1}^n
\lterm_i \subseteq L_i \wedge \bigwedge_{i=1}^n \pto_i \subseteq L_i
\times D_i \Rightarrow \psi(\tup\lterm,\tup\pto)$, where $L_i$'s and
$D_i$'s are finite sets of terms, none of which involves quantified
variables, and $\psi$ is a formula in the theory $T$ with sets and
uninterpreted functions. Moreover, the labeling rule for $\phi \wand
\psi \lab [\tup\lterm,\tup\pto,C]$ uses a fresh constant $d$ that does
not occur in $\phi$ or $\psi$.

\begin{figure}[htb]
\[\begin{array}{c}
\infer[]{\begin{array}{lcl}
    \neg\forall \lterm_1 \forall \lterm_2 & . & 
    \lterm_1 \cup \lterm_2 \subseteq C \cup \terms{\phi*\psi} \Rightarrow \\ 
    && \neg(\lterm_1 \cap \lterm_2 \teq \emptyset \wedge \lterm_1 \cup \lterm_2 \teq \bigcup\tup\lterm \wedge
    \phi\lab[\lterm_1\cap\tup\lterm,\tup\pto,C] \wedge \psi\lab[\lterm_2\cap\tup\lterm,\tup\pto,C])
    \end{array}}{\phi*\psi\lab[\tup\lterm,\tup\pto,C]} 
\\\\ 
\infer[\begin{array}{r}
C' = \bounds{\phi\wedge\psi}{C} \\
d \not\in \terms{\phi \wand \psi} 
\end{array}]{\begin{array}{lcl} 
    \forall \lterm' \forall \pto' & . & \lterm' \subseteq C' \cup \terms{\phi \wand \psi} ~\wedge  \\
    && \pto' \subseteq (C' \cup \terms{\phi \wand \psi}) \times (\terms{\phi \wand \psi}\cup\set{d}) \Rightarrow \\
    && ( \lterm'\cap(\bigcup\tup\lterm) \teq \emptyset ~\wedge~ \phi \lab [\lterm',\pto',C'] ) \Rightarrow \psi \lab [\lterm'\cdot\tup\lterm,\pto'\cdot\tup\pto,C]
\end{array}}{\phi \wand \psi \lab [\tup\lterm,\tup\pto,C]}
\end{array}\]
\caption{Bounded Quantifier Labeling Rules }
\label{fig:bounded-labeling}
\end{figure}


\begin{example}\label{ex:bounded-labeling2}
We revisit below the labeling of the formula $(x \mapsto a \wand y \mapsto b)
\wedge \emp$: 
\[\begin{array}{c}
((x \mapsto a \wand y \mapsto b) \wedge \emp) \lab [\lterm,\pto,C] 
\Longrightarrow^* \\
\lterm\teq \emptyset \wedge \forall \lterm' \subseteq \set{x,y,a,b,c} \forall \pto' \subseteq \set{x,y,a,b,c} \times \set{x,y,a,b,d}) ~.~ \\
\lterm'\cap\lterm\teq \emptyset \wedge \lterm'\teq \set{x} \wedge \ite{x\in\lterm'}{\pto'(x)\teq a}{\top} \wedge x\not\teq\nil \Rightarrow \\
\lterm'\cup\lterm\teq \set{y} \wedge \ite{y\in\lterm'}{\pto'(y)\teq b}{\ite{y\in\lterm}{\pto(y)\teq b}{\top}} \wedge y\not\teq\nil \enspace.\\[-2mm]
\end{array}\]
where $\terms{(x \mapsto a \wand y \mapsto b) \wedge \emp} =
\set{x,y,a,b}$. Observe that the constant $c$ was introduced by the
bounded quantifier labeling of the term $x \mapsto a \wand y \mapsto
b$. \hfill$\blacksquare$
\end{example}

The next lemma states the soundness of the translation of $\seplog(T)$
formulae in a fragment of $T$ that contains only bounded quantifiers,
by means of the rules in Figure \ref{fig:bounded-labeling}. 

\begin{lemma}\label{lemma:bounded-labeling}
  Given a formula $\varphi$ in the language $\seplog(T)$, for any
  interpretation $\I$ of $T$, let $L \subseteq \locs^\I \setminus
  \terms{\varphi}^\I$ be a set of locations such that
  $\card{L}=\len{\varphi}$ and $v \in \data^\I \setminus
  \terms{\varphi}^\I$ be a data value. Then there exists a heap $h$
  such that $\I,h \models_\seplog \varphi$ iff there exist heaps
  $\tup{h'}=\tuple{h'_1, \ldots, h'_n}$ and
  $\tup{h''}=\tuple{h''_1,\ldots,h''_n}$ such
  that: \begin{compactenum}
  \item for all $1 \leq i < j \leq n$, we have $h'_i \# h'_j$,
  \item for all $1 \leq i \leq n$, we have $h'_i \subseteq h''_i$ and
  \item \(\I[\tup\lterm \leftarrow \dom(\tup{h'})][\tup\pto \leftarrow
    \tup{h''}][C \leftarrow L][d \leftarrow v] \models_T \labd{\varphi
    \lab [\tup\lterm,\tup\pto,C]}\enspace.\)
  \end{compactenum}
\end{lemma}

\section{A Counterexample-Guided Approach for Solving $\seplogl{T}$ Inputs}
\label{sec:solve-ceg}

This section presents a novel decision procedure for the $(\seplog,
T)$-satisfiability of the set of quantifier-free $\seplogl{T}$ formulae $\varphi$.  To this
end, we present an efficient decision procedure for the
$T$-satisfiability of $\labd{(\varphi \lab [\lterm,\pto,C])}$, obtained as
the result of the transformation described in
Section~\ref{sec:boundq}. The main challenge in doing so is treating
the universal quantification occurring in $\labd{(\varphi \lab [\lterm,\pto,C])}$.
As mentioned, the key to decidability is that all quantified formulae
in $\labd{(\varphi \lab [\lterm,\pto,C])}$ are equivalent to formulas
of the form $\forall \vec \qvar. ( \bigwedge \vec \qvar \subseteq \vec
s ) \Rightarrow \varphi$, where each term in the tuple $\vec s$ is a
finite set (or product of sets) of ground $\Sigma$-terms. For brevity,
we write $\forall \vec \qvar \subseteq \vec s. \varphi$ to denote a
quantified formula of this form.  While such formulae are clearly
equivalent to a finite conjunction of instances, the cost of
constructing these instances is in practice prohibitively expensive.
Following recent approaches for handling universal
quantification~\cite{GeDeM-CAV-09,ReynoldsDKBT15Cav,bjornerplaying,RKK2015},
we use a counterexample-guided approach for choosing instances of
quantified formulae that are relevant to the satisfiability of our
input.  The approach is based on an iterative procedure maintaining an
evolving set of quantifier-free $\Sigma$-formulae $\Gamma$, which is
initially a set of formulae obtained from $\varphi$ by a purification
step, described next.

We associate each closed quantified formula a boolean variable $A$,
called the \emph{guard} of $\forall \vec \qvar. \varphi$, and a
(unique) set of Skolem symbols $\vec k$ of the same sort as $\vec
\qvar$.  We write $(A, \vec k ) \leftrightharpoons \forall \vec
\qvar. \varphi$ to denote that $A$ and $\vec k$ are
associated with $\forall \vec \qvar. \varphi$.
For a set of formulae $\Gamma$, we write $\quants{\Gamma}$ to
denote the set of quantified formulae whose guard occurs within a formula in $\Gamma$.
We write $\purify{ \psi }$ for the result of replacing in $\psi$ all
closed quantified formulae (not occurring beneath other quantifiers in $\psi$) with their corresponding guards.
Conversely, we write $\unpurify{\Gamma}$ to denote the result of replacing
all guards in $\Gamma$ by the quantified formulae they are associated with. 
We write $\purifyrec{ \psi }$ denote the (smallest) set of $\Sigma$-formulae where:

\[\begin{array}{r@{\hspace{1em}}r}
\purify{\psi} \in \purifyrec{ \psi } \\[.2ex]
( \neg A \Rightarrow \purify{ \neg \varphi[ \vec k/\vec \qvar] } ) \in \purifyrec{ \psi } & \text{if } \forall \vec \qvar. \varphi \in \quants{ \purifyrec{ \psi } } 
\text{ where } ( A, \vec k ) \leftrightharpoons \forall \vec \qvar. \varphi. \\[.2ex]
\end{array}\]

\noindent
It is easy to see that if $\psi$ is a $\Sigma$-formula possibly containing quantifiers, then 
$\purifyrec{ \psi }$ is a set of quantifier-free $\Sigma$-formulae, and
if all quantified formulas in $\psi$ are of the form 
$\forall \vec \qvar \subseteq \vec s. \varphi$ mentioned above, 
then all quantified formulas in $\quants{ \purifyrec{ \psi } }$ are also of this form.

\begin{example}
If $\psi$ is the formula $\forall x. ( P( x ) \Rightarrow \neg \forall y. R( x, y ) )$,
then $\purifyrec{ \psi }$ is the set:
\[\begin{array}{l}
\{ A_1, \neg A_1 \Rightarrow \neg( P( k_1 ) \Rightarrow A_2 ), \neg A_2 \Rightarrow \neg R( k_1, k_2 ) \}
\end{array}\]
where $( A_1, k_1 ) \leftrightharpoons \forall x. ( P( x ) \Rightarrow \neg \forall y. R( x, y ) )$
and $( A_2, k_2 ) \leftrightharpoons \forall y. R( k_1, y )$.
\hfill$\blacksquare$
\end{example}

\begin{figure}[t]
\begin{framed}
$\solvesl{T}( \varphi )$:
\begin{enumerate}
\item[\ ] Let $C$ be a set of fresh constants of sort $\locs$ such that $\card{C}=\len{\varphi}$.
\item[\ ] Let $\lterm$ and $\pto$ be a fresh symbols of sort $\sets(\locs)$ and $\locs \Rightarrow \data$ respectively.
\item[\ ] Return $\solveslt{T}( \purifyrec{ \labd{(\varphi \lab [\lterm,\pto,C])} } )$.
\end{enumerate}
$\solveslt{T}( \Gamma )$:
\begin{enumerate}
\item If $\Gamma$ is $T$-unsatisfiable,
\begin{enumerate}
\item[\ ] return ``unsat",
\end{enumerate}
\item[\ ] else let $\I$ be a $T$-model of $\Gamma$.
\item If $\Gamma, A \models_T \purify{ \psi[ \vec k/\vec \qvar] }$ for all $\forall \vec \qvar. \psi \in \quants{\Gamma}$,
where $(A, \vec k) \leftrightharpoons \forall \vec \qvar. \psi$ and $A^\I = \ltrue$, 
\begin{enumerate}
\item[\ ] return ``sat",
\end{enumerate}
\item[\ ]\begin{tabular}{l@{\hspace{0.2em}}l@{\hspace{0.2em}}r}
else & let $\J$ be a $T$-model of $\Gamma \cup \{ A, \neg \purify{ \psi[ \vec k/\vec \qvar] } \}$ for some $\ord{\Gamma}{\I}$-minimal
 $\forall \vec \qvar \subseteq \vec s. \psi$, \\
 &  where $(A, \vec k ) \leftrightharpoons \forall \vec \qvar \subseteq \vec s. \psi$.
\end{tabular}
\item Let $\vec t$ be a vector of terms, such that $\vec t \subseteq \vec s$, and $\vec t^\J = \vec k^\J$.
\item[\ ] Return $\solveslt{T}$( $\Gamma \cup \purifyrec{ A \Rightarrow \psi[ \vec t/ \vec \qvar ]} )$.
\end{enumerate}
\vspace*{-2ex}
\end{framed}
\vspace*{-2ex}
\caption{Procedure $\solvesl{T}$ for deciding $(\seplog,T)$-satisfiability of $\seplogl{T}$ formula $\varphi$.
\label{fig:solve}}
\vspace*{-\baselineskip}
\end{figure}

Our algorithm $\solvesl{T}$
for determining the $(\seplog,T)$-satisfiability of input
$\varphi$ is given in Figure~\ref{fig:solve}.
It first constructs the set $C$ based on the value of $\len{\varphi}$,
which it computes by traversing the structure of $\varphi$.
It then invokes the subprocedure $\solveslt{T}$
on the set $\purifyrec{ \labd{(\varphi \lab [\lterm,\pto,C])} }$ where $\lterm$ and $\pto$ are fresh free symbols.  

At a high level, the recursive procedure $\solveslt{T}$ takes as input
a (quantifier-free) set of $T$-formulae $\Gamma$, where $\Gamma$ is
$T$-unsatisfiable if and only if $\labd{(\varphi \lab [\lterm,\pto,C])}$ is. 
On each invocation,
$\solveslt{T}$ will either\begin{inparaenum}[(i)]
\item terminate with ``unsat", in which case our
input $\varphi$ is $T$-unsatisfiable, 
\item terminate with ``sat", in which
case our input $\varphi$ is $T$-satisfiable, or 
\item add the set
corresponding to the purification of the instance $\purifyrec{ A
  \Rightarrow \psi[ \vec t/ \vec \qvar ]}$ to $\Gamma$ and repeats.
\end{inparaenum}

In more detail, in Step 1 of the procedure,
we determine the $T$-satisfiability of $\Gamma$
using a combination of a satisfiability solver and a decision procedure for $T$
\footnote{Non-constant Skolem symbols $k$ introduced by the procedure 
may be treated as uninterpreted functions.
Constraints of the form $k \subseteq S_1 \times S_2$ are translated to
$\bigwedge_{c \in S_1} k( c ) \in S_2$.
Furthermore, the domain of $k$ may be restricted to the set $\{ c^\I \mid c \in S_1 \}$ in models $\I$
found in Steps 1 and 2 of the procedure. 
This restriction comes with no loss of generality since,
by construction of $\labd{(\varphi \lab [\lterm,\pto,C])}$, 
$k$ is applied only to terms occurring in $S_1$.
}.
If $\Gamma$ is $T$-unsatisfiable, since $\Gamma$ is $T$-entailed by $\unpurify{\Gamma}$,
we may terminate with ``unsat".
Otherwise, there is a $T$-model $\I$ for $\Gamma$ and $T$.
In Step 2 of the procedure,
for each $A$ that is interpreted to be true by $\I$,
we check whether 
$\Gamma \cup \set{A}$ $T$-entails $\purify{ \psi[ \vec k/\vec \qvar] }$ for fresh free constants $\vec k$,
which can be accomplished by determining whether
$\Gamma \cup \{ A, \neg \purify{ \psi[ \vec k/\vec \qvar] } \}$ is $T$-unsatisfiable.
If this check succeeds for a quantified formula $\forall \vec \qvar. \psi$,
the algorithm has established that $\forall \vec \qvar. \psi$ is entailed by $\Gamma$.
If this check succeeds for all such quantified formulae, 
then $\Gamma$ is equivalent to $\unpurify{\Gamma}$, and
we may terminate with ``sat".
Otherwise, let $Q^+_\I( \Gamma )$ be the subset of $Q( \Gamma )$ 
for which this check did not succeed.
We call this the set of \emph{active quantified formulae} for $(\I, \Gamma)$.
We consider an active quantified formula that is minimal with respect to the 
relation $\ord{\Gamma}{\I}$ over $\quants{ \Gamma }$, where:
\[\begin{array}{r@{\hspace{1em}}r}
\varphi \ord{\Gamma}{\I} \psi & \text{if and only if } 
\varphi \in \quants{ \purifyrec{ \psi } } \cap Q^+_\I( \Gamma ) \\
\end{array}\]

By this ordering, our approach considers innermost active quantified formulae first.
Let $\forall \vec \qvar. \psi$ be minimal with respect to $\ord{\Gamma}{\I}$,
where $(A, \vec k) \leftrightharpoons \forall \vec \qvar. \psi$.
Since $\Gamma, A$ does not $T$-entail $\purify{ \psi[ \vec k/\vec \qvar] }$,
there must exist a model $\J$ for
$\Gamma \cup \{ \purify{ \neg \psi[ \vec k/\vec \qvar] } \}$ where $A^\J = \top$.
In Step 3 of the procedure,
we choose a tuple of terms $\vec t = ( t_1, \ldots, t_n )$ based on the model $\J$,
and add to $\Gamma$ the set of formulae obtained by purifying $A \Rightarrow \psi[ \vec t/ \vec \qvar ]$,
where $A$ is the guard of $\forall \vec \qvar \subseteq \vec s. \psi$.
Assume that $\vec s = ( s_1, \ldots, s_n )$ and recall that each $s_i$ is a finite union of ground $\Sigma$-terms.
We choose each $\vec t$ such that $t_i$ is a subset of $s_i$ for each $i = 1, \ldots n$,
and $\vec t^\J = \vec k^\J$.  
These two criteria are the key to the termination of the algorithm:
the former ensures that only a finite number of possible instances can ever be added to $\Gamma$,
and the latter ensures that we never add the same instance more than once.


\begin{theorem}
\label{thm:solveslt}
For all $\seplogl{T}$ formulae $\varphi$,
$\solvesl{T}( \varphi )$:
\begin{enumerate}
\item \label{it:sltunsat} Answers ``unsat" only if $\varphi$ is $(\seplog,T)$-unsatisfiable.
\item \label{it:sltsat} Answers ``sat" only if $\varphi$ is $(\seplog,T)$-satisfiable.
\item \label{it:sltterm} Terminates.
\end{enumerate}
\end{theorem}
By Theorem~\ref{thm:solveslt}, $\solvesl{T}$ is a decision procedure
for the $(\seplog,T)$-satisfiability of the language of
quantifier-free $\seplogl{T}$ formulae. The following corollary gives
a tight complexity bound for the $(\seplog,T)$-satisfiability problem.

\begin{corollary}\label{cor:slt-pspace}
  The $(\seplog,T)$-satisfiability problem is PSPACE-complete for any
  theory $T$ whose satisfiability (for the quantifier-free fragment)
  is in PSPACE. 
\end{corollary}

In addition to being sound and complete, in practice, the approach
$\solvesl{T}$ terminates in much less time that its theoretical
worst-case complexity, given by the above corollary. This fact is
corroborated by our evaluation of our prototype implementation of the
algorithm, described in Section~\ref{sec:eval}, and in the following
examples.

\begin{example}\label{ex:slsolve3}
Consider the $\seplog(T)$ formula $\varphi \equiv \emp \wedge (y \mapsto 0 \wand y \mapsto 1) \wedge y \not\teq \nil$.
When running
$\solvesl{T}(\varphi)$, we first compute the set $C = \{ c \}$,
and introduce fresh symbols $\lterm$ and $\pto$ of sorts $\sets(\locs)$ and $\locs \rightarrow \data$ respectively.
The formula $\labd{(\varphi \lab [\lterm,\pto,C])}$ is
$\lterm \teq \emptyset \wedge \forall \lterm_4 \forall \pto'. \psi \wedge y \not\teq \nil$, where after simplification $\psi$ is:
\[\begin{array}{r@{\hspace{.5em}}c@{\hspace{.5em}}l}
 \psi_4 & \equiv & 
( \lterm_4 \subseteq \set{y,0,1,c} \wedge \pto' \subseteq \set{y,0,1,c} \times \set{y,0,1,d})\Rightarrow \\
 & & ( \lterm_4\cap\lterm \teq \emptyset \wedge \lterm_4 \teq \set{y} \wedge \pto'(y) \teq 0 \wedge y\not\teq\nil) \Rightarrow \\
 & & ( \lterm_4\cup\lterm \teq \set{y} \wedge \ite{y\in\lterm_4}{\pto'(y) \teq 1}{\pto(y) \teq 1} \wedge y\not\teq\nil ) \enspace\\[-2mm]
\end{array}\]
Let $( A_4, (k_1,k_2) ) \leftrightharpoons \forall \lterm_4 \forall \pto'.\psi_4$. 
We call the subprocedure $\solveslt{T}$ on $\Gamma_0$, where:
\[\begin{array}{r@{\hspace{.25em}}c@{\hspace{.25em}}r@{\hspace{.25em}}c@{\hspace{.25em}}r}
\Gamma_0 & \equiv & \purifyrec{\labd{(\varphi \lab [\lterm,\pto,C])}} & \equiv &
\{ \lterm \teq \emptyset \wedge A_4 \wedge y \not\teq \nil, 
\neg A_4 \Rightarrow \neg \psi_4[ k_1, k_2 / \lterm_4, \pto' ] \}.
\end{array}\]
The set $\Gamma_0$ is $T$-satisfiable with a model $\I_0$ where $A_4^{\I_0} = \top$.
Step 2 of the procedure determines a model $\J$ for $\Gamma_0 \cup \{ A_4, \neg \psi_4[ k_1, k_2 / \lterm_4, \pto' ] \}$. 

Let $t_1$ be $\{ y \}$, where we know $t_1^\J = k_1^\J$ since $\J$ must satisfy $k_1 \teq \{ y \}$ as a consequence of $\neg \psi_4[ k_1, k_2 / \lterm_4, \pto' ]$.
Let $t_2$ be a well-sorted subset of $\set{y,0,1,c} \times \set{y,0,1,d}$ such that $t_2^\J = k_2^\J$.
Such a subset exists since $\J$ satisfies $k_2 \subseteq \set{y,0,1,c} \times \set{y,0,1,d}$.
Notice that $t_2( y )^\J = 0^\J$ since $\J$ must satisfy $k_2( y ) \teq 0$.
Step 3 of the procedure 
recursively invokes $\solveslt{T}$ on $\Gamma_1$, where:
\[\begin{array}{r@{\hspace{.25em}}c@{\hspace{.25em}}l}
\Gamma_1  & \equiv & \Gamma_0 \cup \purifyrec{ A_4 \Rightarrow \psi_4[ t_1, t_2/\lterm_4, \pto' ] } \\
& \equiv & \Gamma_0 \cup \{ A_4 \Rightarrow y\not\teq\nil \Rightarrow 
( \{ y \} \teq \{ y \} \wedge \ite{y\in\{y\}}{0 \teq 1}{\pto(y) \teq 1} \wedge y \not\teq \nil ) \} \\
& \equiv & \Gamma_0 \cup \{ A_4 \Rightarrow y\not\teq\nil \Rightarrow \bot \}
\end{array}\]
The set $\Gamma_1$ is $T$-unsatisfiable, since the added constraint contradicts $A_4 \wedge y \not\teq \nil$.
\hfill$\blacksquare$
\end{example}

\subsection{Integration in DPLL($T$)}

We have implemented the algorithm described in this section within the SMT solver CVC4~\cite{CVC4-CAV-11}.
Our implementation accepts 
an extended syntax of SMT lib version 2 format~\cite{smtlib25} for specifying $\seplogl{T}$ formulae.
In contrast to the presentation so far,
our implementation does not explicitly introduce quantifiers,
and instead treats $\seplog$ atoms natively using an integrated subsolver that 
expands the semantics of these atoms in lazy fashion.

In more detail, given a $\seplogl{T}$ input $\varphi$, our
implementation lazily computes the expansion of $\labd{(\varphi \lab
  [\lterm,\pto,C])}$ based on the translation rules in
Figures~\ref{fig:labeling-rules} and \ref{fig:bounded-labeling} and
the counterexample-guided instantiation procedure in
Figure~\ref{fig:solve}. This is accomplished by a module, which we
refer to as the $\seplog$ solver, that behaves analogously to a
DPLL($T$)-style \emph{theory solver}, that is, a dedicated solver
specialized for the $T$-satisfiability of a conjunction of
$T$-constraints~\footnote{Strictly speaking, the $\seplog$ functions
  $\mapsto, *, \wand$ are not symbols belonging to a first-order
  theory, and thus this module is not a theory solver in the standard
  sense.}.

The DPLL($T$) solving architecture~\cite{GanzingerHagenNieuwenhuisOliverasTinelli04} used by most modern SMT solvers,
given as input a set of quantifier-free $T$-formulae $\Gamma$,
incrementally constructs of set of literals over the atoms of $\Gamma$
until either it finds a set $M$ that entail $\Gamma$ at the propositional level,
or determines that such a set cannot be found.
In the former case, we refer to $M$ as a \emph{satisfying assignment} for $\Gamma$.
If $T$ is a combination of theories $T_1 \cup \ldots \cup T_n$, 
then $M$ is partitioned into $M_1 \cup \ldots \cup M_n$ where the atoms of $M_i$ are 
either $T_i$-constraints or (dis)equalities shared over multiple theories.
We use a theory solver (for $T_i$) to determine the $T_i$-satisfiability of the set $M_i$,
interpreted as a conjunction.
Given $M_i$, the solver will either add additional formulae to $\Gamma$,
or otherwise report that $M_i$ is $T_i$-satisfiable.

For $\seplogl{T}$ inputs,
we extend our input syntax with a set of functions:
\[\begin{array}{l@{\hspace{1em}}l@{\hspace{1em}}l}
\mapsto : \locs \times \data \rightarrow \bools & *^n : \bools^n \rightarrow \bools & \emp : \bools \\
\wand : \bools \times \bools \rightarrow \bools & \lblf : \bools \times \sets(\locs) \rightarrow \bools \\
\end{array}\]
which we call \emph{spatial functions}
\footnote{ These functions are over the $\bools$ sort.  We refer to
  these functions as taking \emph{formulae} as input, where formulae
  may be cast to terms of sort $\bools$ through use of an if-then-else
  construct.}.  We refer to $\lblf$ as the \emph{labeling predicate},
which can be understood as a placeholder for the $\lab$ transformation
in Figures~\ref{fig:labeling-rules} and \ref{fig:bounded-labeling}.
We refer to $p( \vec t )$ as an \emph{unlabeled spatial atom} if $p$
is one of $\{ \emp, \mapsto, *^n, \wand \}$ and $\vec t$ is a vector
of terms not containing $\lblf$.  If $a$ is an unlabeled spatial atom,
We refer to $\lblf( a, \lterm )$ as a \emph{labeled spatial atom}, and
extend these terminologies to literals.  We assume that all
occurrences of spatial functions in our input $\varphi$ occur only in
unlabeled spatial atoms. Moreover, during execution, our
implementation transforms all spatial atoms into a \emph{normal form},
by applying associativity to flatten nested applications of $*$, and
distributing $\Sigma$-formulae over spatial connectives, e.g.\ $(
( \xterm \mapsto \yterm \wedge \tterm \teq \uterm ) * \zterm \mapsto
\wterm ) \iff \tterm \teq \uterm \wedge ( \xterm \mapsto \yterm *
\zterm \mapsto \wterm )$.

When constructing satisfying assignments for $\varphi$, we relegate
the set of all spatial literals $M_k$ to the $\seplog$ solver.  For
all unlabeled spatial literals $(\neg) a$, we add to $\Gamma$ the
formula $(a \Leftrightarrow \lblf( a, \lterm_0 ) )$, where $\lterm_0$
is a distinguished free constant of sort $\sets(\locs)$.  Henceforth,
it suffices for the $\seplog$ solver to only consider the labeled
spatial literals in $M_k$.  To do so, firstly, it adds to $\Gamma$
formulae based on the following criteria, which model one step of the
reduction from Figure~\ref{fig:labeling-rules}:
\[\begin{array}{l@{\hspace{1em}}l}
\lblf( \emp, \lterm ) \Leftrightarrow \lterm \teq \emptyset
& \text{if } (\neg) \lblf( \emp, \lterm ) \in M_k \\
\lblf( t \mapsto u, \lterm ) \Leftrightarrow \lterm \teq \{ t \} \wedge \pto( t ) \teq u \wedge t \not\teq \nil
& \text{if } (\neg)  \lblf( t \mapsto u, \lterm ) \in M_k \\
\lblf( ( \varphi_1 * \ldots * \varphi_n ), \lterm ) \Rightarrow ( \varphi_1[ \ell_1 ] \wedge \ldots \wedge \varphi_n[\ell_n] )
& \text{if } \lblf( ( \varphi_1 * \ldots * \varphi_n ), \lterm ) \in M_k \\
\neg \lblf( ( \varphi_1 \wand \varphi_2 ), \lterm ) \Rightarrow ( \varphi_1[ \ell_1 ] \wedge \neg \varphi_2[\ell_2] )
& \text{if } \neg \lblf( ( \varphi_1 \wand \varphi_2 ), \lterm ) \in M_k 
\end{array}\]
where each $\ell_i$ is a fresh free constant, and 
$\varphi_i[ \ell_i ]$ denotes the result of replacing each spatial atom $a$ in $\varphi_i$ with $\lblf( a, \ell_i )$.
These formulae are added eagerly when such literals are added to $M_k$.
To handle negated $*$-atoms and positive $\wand$-atoms,
the $\seplog$ solver adds to $\Gamma$ formulae based on the criteria:
\[\begin{array}{l@{\hspace{1em}}l}
\neg \lblf( ( \varphi_1 * \ldots * \varphi_n ), \lterm ) \Rightarrow ( \neg \varphi_1[ t_1 ] \vee \ldots \vee \neg \varphi_n[t_n] )
& \text{if } \neg \lblf( ( \varphi_1 * \ldots * \varphi_n ), \lterm ) \in M_k \\
\lblf( ( \varphi_1 \wand \varphi_2 ), \lterm ) \Rightarrow ( \neg \varphi_1[ t_1, f_1 ] \vee \varphi_2[t_2, f_2] )
& \text{if } \lblf( ( \varphi_1 \wand \varphi_2 ), \lterm ) \in M_k
\end{array}\]
where each $t_i$ and $f_i$ is chosen based on the same criterion as described in Figure~\ref{fig:solve}.
For wand, we write $\varphi_i[ t_i, f_i ]$ to denote $\varphi_i'[ t_i ]$, where
$\varphi_i'$ is the result of replacing all atoms of the form $t \mapsto u$ where $t \in t_1$ in $\varphi_i$ by $f_i( t ) \teq u$.

CVC4 uses a scheme for incrementally checking the $T$-entailments
required by $\solveslt{T}$, as well as constructing models $\J$
satisfying the negated form of the literals in literals in $M_k$
before choosing such terms~\footnote{For details, see Section 5
  of~\cite{RKK2015}.}.  The formula of the above form are added to
$\Gamma$ lazily, that is, after all other solvers (for theories $T_i$)
have determined their corresponding sets of literals $M_i$ are
$T_i$-satisfiable.

\paragraph{\bf Partial Support for Quantifiers} 
In many practical cases it is useful to check the validity of
entailments between existentially quantified $\seplog(T)$ formulae
such as $\exists \vec{x} ~.~ \phi(\vec{x})$ and $\exists \vec{y} ~.~
\psi(\vec{y})$. Typically, this problem translates into a
satisfiability query for an $\seplog(T)$ formula \(\exists \vec{x}
\forall \vec{y} ~.~ \phi(\vec{x}) \wedge \neg \psi(\vec{y})\), with
one quantifier alternation. A partial solution to this problem is to
first check the satisfiability of $\phi$. If $\phi$ is not
satisfiable, the entailment holds trivially, so let us assume that
$\phi$ has a model. Second, we check satisfiability of $\phi \wedge
\psi$. Again, if this is unsatisfiable, then the entailment cannot
hold, because there exists a model of $\phi$ which is not a model of
$\psi$. Else, if $\phi \wedge \psi$ has a model, we add an equality
$x=y$ for each pair of variables $(x,y) \in \vec{x}\times\vec{y}$ that
are mapped to the same term in this model, the result being a
conjunction $E(\vec{x},\vec{y})$ of equalities. Finally, we check the
satisfiability of the formula $\phi \wedge \neg\psi \wedge E$. If this
formula is unsatisfiable, the entailment is valid, otherwise, the test
is inconclusive.

\section{Evaluation}
\label{sec:eval}

We tested our implementation of the $(\seplog,T)$-satisfiability
procedure in CVC4 (version 1.5 prerelease) on two kinds of benchmarks: \begin{inparaenum}[(i)]
\item finite unfoldings of inductive predicates with data constraints,
  mostly inspired by existing benchmarks, such as SL-COMP'14
  \cite{sl-comp14}, and
\item verification conditions automatically generated by applying the
  weakest precondition calculus of ~\cite{IshtiaqOHearn01} to the
  program loops in Figure \ref{fig:loops} several times.
\end{inparaenum}
All experiments were run on a 2.80GHz Intel(R) Core(TM) i7 CPU machine
with with 8MB of cache\footnote{The CVC4 binary and examples used in
  these experiments are available at
  \url{http://homepage.cs.uiowa.edu/~ajreynol/ATVA2016-seplog/}.}.
For a majority of benchmarks, the runtime of CVC4 is quite low, with
the exception of the $n=4,8$ cases of the entailments between
$\mathsf{tree}_1^n$ and $\mathsf{tree}_2^n$ formulae, which resulted
in a timeout after $300$ seconds.

\begin{figure}[htb]
\begin{center}
\begin{minipage}{7cm}
{\footnotesize\begin{tabbing}
1:  \texttt{wh}\=\texttt{ile} $\wterm \neq \nil$ \texttt{do} \\
2:  \> $\mathbf{assert}(\wterm.\datap = 0)$ \\
3:  \> $\vterm := \wterm$; \\
4:  \> $\wterm := \wterm.\nextp$; \\
5:  \> $\mathbf{dispose}(\vterm)$;   \\
6:  \> \texttt{do} \\
\\
\> {\bf (z)disp}
\end{tabbing}}
\end{minipage}
\hspace*{2cm}
\begin{minipage}{6.5cm}
{\footnotesize\begin{tabbing}
1:  \texttt{wh}\=\texttt{ile} $\uterm \neq \nil$ \texttt{do} \\
2:  \> $\mathbf{assert}(\uterm.\datap = 0)$ \\
3:  \> $\wterm := \uterm.\nextp$; \\
4:  \> $\uterm.\nextp := \vterm$; \\
5:  \> $\vterm := \uterm$;   \\
6:  \> $\uterm := \wterm$;   \\
7:  \> \texttt{do} \\
\> {\bf (z)rev}
\end{tabbing}}
\end{minipage}
\begin{tabular}{r@{ }c@{ }l@{\hskip 1.5cm}r@{ }c@{ }l}
$\mathsf{ls}^0(x)$ & $\triangleq$ & $\emp \land x = \nil$ & $\mathsf{zls}^0(x)$ & $\triangleq$ & $\emp \land x = \nil$\\
$\mathsf{ls}^n(x)$ & $\triangleq$ & $\exists y \, . \, x \mapsto y * \mathsf{ls}^{n-1}(y)$ & $\mathsf{zls}^n(x)$ & $\triangleq$ & $\exists y \, . \, x \mapsto (0, y) * \mathsf{zls}^{n-1}(y)$
\end{tabular}
\end{center}
\vspace*{-\baselineskip}
\caption{Program Loops}\label{fig:loops}
\vspace*{-\baselineskip}
\end{figure}

\begin{table}[htb]
\begin{tabular}{|c|c|c|c|c|c|c|}
\hline
{\bf lhs} & {\bf rhs} & $n=1$ & $n=2$ & $n=3$ & $n=4$ & $n=8$ \\
\hline
\multicolumn{7}{|c|}{\bf Unfoldings of inductive predicates} \\
\hline
$\scriptstyle \mathsf{lseg_1}(x, y, a) \triangleq \emp \land x = y \lor \exists z \exists b \, . \,$ & $\scriptstyle \mathsf{lseg_2}(x, y, a) \triangleq \emp \land x = y \lor \exists z \exists b \, . \,$ & {\scriptsize unsat} & {\scriptsize unsat} & {\scriptsize unsat} & {\scriptsize unsat} & {\scriptsize unsat}  \\[-3pt]
$\scriptstyle x \mapsto (a, z) * \mathsf{lseg_1}(z, y, b) \land b = a+10$ & $\scriptstyle x \mapsto (a, z) * \mathsf{lseg_2}(z, y, b) \land a \leq b$ & {\scriptsize $<0.01$s} & {\scriptsize $<0.01$s} & {\scriptsize $<0.01$s} & {\scriptsize 0.01s} & {\scriptsize 0.01s}  \\
\hline
$\scriptstyle \mathsf{tree_1}(x, a) \triangleq \emp \land x = \nil \lor \exists y \exists z \exists b \exists c \, . \, $ & $\scriptstyle \mathsf{tree_2}(x, a) \triangleq \emp \land x = \nil \lor \exists y \exists z \exists b \exists c \, . \, $ & {\scriptsize unsat} & {\scriptsize unsat} & {\scriptsize unsat} & {\scriptsize timeout} & {\scriptsize timeout} \\[-3pt]
$\scriptstyle x \mapsto (a, y, z) * \mathsf{tree_1} (y, b) * \mathsf{tree_1}(z, c)  \land$ & $\scriptstyle x \mapsto (a, y, z) * \mathsf{tree_2} (y, b) * \mathsf{tree_2}(z, c)  \land$ & {\scriptsize $<0.01$s} & {\scriptsize 0.06s} & {\scriptsize 1.89s} & {\scriptsize $>300$s} & {\scriptsize $>300$s}  \\[-3pt]
$\scriptstyle b = a - 10 \land c = a + 10$ & $\scriptstyle b \leq a \land a \leq c$ & & & & &  \\
\hline
$\scriptstyle \mathsf{pos_1}(x, a) \triangleq x \mapsto a \lor \exists y \exists b \, . \,$ & $\scriptstyle \mathsf{neg_1}(x, a) \triangleq \lnot x \mapsto a \lor \exists y \exists b \, . \,$ & {\scriptsize unsat} & {\scriptsize unsat} & {\scriptsize unsat} & {\scriptsize unsat} & {\scriptsize unsat} \\[-3pt]
$\scriptstyle x \mapsto a * \mathsf{pos_1}(y, b)$ & $\scriptstyle x \mapsto a * \mathsf{neg_1}(y, b)$ & {\scriptsize 0.02s} & {\scriptsize 0.04s} & {\scriptsize 0.11s} & {\scriptsize 0.25 s} & {\scriptsize 3.01s} \\
\hline
$\scriptstyle \mathsf{pos_1}(x, a) \triangleq x \mapsto a \lor \exists y \exists b \, . \,$ & $\scriptstyle \mathsf{neg_2}(x, a) \triangleq x \mapsto a \lor \exists y \exists b \, . \,$ & {\scriptsize unsat} & {\scriptsize unsat} & {\scriptsize unsat} & {\scriptsize unsat} & {\scriptsize unsat}  \\[-3pt]
$\scriptstyle x \mapsto a * \mathsf{pos_1}(y, b)$ & $\scriptstyle \lnot x \mapsto a * \mathsf{neg_2}(y, b)$ & {\scriptsize 0.01s} & {\scriptsize 0.05s} & {\scriptsize 0.11s} & {\scriptsize 0.23s} & {\scriptsize 2.10s}  \\
\hline
$\scriptstyle \mathsf{pos_1}(x, a) \triangleq x \mapsto a \lor \exists y \exists b \, . \,$ & $\scriptstyle \mathsf{neg_3}(x, a) \triangleq x \mapsto a \lor \exists y \exists b \, . \,$ & {\scriptsize unsat} & {\scriptsize unsat} & {\scriptsize unsat} & {\scriptsize unsat} & {\scriptsize unsat}  \\[-3pt]
$\scriptstyle x \mapsto a * \mathsf{pos_1}(y, b)$ & $\scriptstyle x \mapsto a * \lnot \mathsf{neg_3}(y, b)$ & {\scriptsize 0.02s} & {\scriptsize 0.07s} & {\scriptsize 0.24s} & {\scriptsize 0.46s} & {\scriptsize 4.05s}  \\
\hline
$\scriptstyle \mathsf{pos_1}(x, a) \triangleq x \mapsto a \lor \exists y \exists b \, . \,$ & $\scriptstyle \mathsf{neg_4}(x, a) \triangleq x \mapsto a \lor \exists y \exists b \, . \,$ & {\scriptsize unsat} & {\scriptsize sat} & {\scriptsize unsat} & {\scriptsize sat} & {\scriptsize sat} \\[-3pt]
$\scriptstyle x \mapsto a * \mathsf{pos_1}(y, b)$ & $\scriptstyle \lnot x \mapsto a * \lnot \mathsf{neg_4}(y, b)$ & {\scriptsize 0.05s} & {\scriptsize 0.24s} & {\scriptsize 0.33s} & {\scriptsize 2.77s} & {\scriptsize 24.72s} \\
\hline
$\scriptstyle \mathsf{pos_2}(x, a) \triangleq x \mapsto a \lor \exists y \, . \,$ & $\scriptstyle \mathsf{neg_5}(x, a) \triangleq \lnot x \mapsto a \lor \exists y \, . \,$ & {\scriptsize unsat} & {\scriptsize unsat} & {\scriptsize unsat} & {\scriptsize unsat} & {\scriptsize unsat} \\[-3pt]
$\scriptstyle x \mapsto a * \mathsf{pos_2}(a, y)$ & $\scriptstyle x \mapsto a * \mathsf{neg_5}(a, y)$ & {\scriptsize 0.02s} & {\scriptsize 0.05s} & {\scriptsize 0.14s} & {\scriptsize 0.32s} & {\scriptsize 3.69s} \\
\hline
$\scriptstyle \mathsf{pos_2}(x, a) \triangleq x \mapsto a \lor \exists y \, . \,$ & $\scriptstyle \mathsf{neg_6}(x, a) \triangleq x \mapsto a \lor \exists y \, . \,$ & {\scriptsize sat} & {\scriptsize unsat} & {\scriptsize unsat} & {\scriptsize unsat} & {\scriptsize unsat} \\[-3pt]
$\scriptstyle x \mapsto a * \mathsf{pos_2}(a, y)$ & $\scriptstyle \lnot x \mapsto a * \mathsf{neg_6}(a, y)$ & {\scriptsize 0.02s} & {\scriptsize 0.04s} & {\scriptsize 0.13s} & {\scriptsize 0.27s} & {\scriptsize 2.22s}  \\
\hline
\multicolumn{7}{|c|}{\bf Verification conditions} \\
\hline
$\scriptstyle \mathsf{ls}^{n}(w)$ & $\scriptstyle\mathsf{wp}(\mathbf{disp}, \mathsf{ls}^{n-1}(w))$ & {\scriptsize $<0.01$s} & {\scriptsize 0.02s} & {\scriptsize 0.05s} & {\scriptsize 0.12s} & {\scriptsize 1.97s} \\
\hline
$\scriptstyle\mathsf{ls}^{n}(w)$ & $\scriptstyle\mathsf{wp}^{n}(\mathbf{disp},\emp \land w = \nil)$ & {\scriptsize $<0.01$s} & {\scriptsize 0.02s} & {\scriptsize 0.12s} & {\scriptsize 0.41s} & {\scriptsize 22.97s} \\
\hline
$\scriptstyle\mathsf{zls}^{n}(w)$  & $\scriptstyle\mathsf{wp}(\textbf{zdisp}, \mathsf{zls}^{n-1}(w))$  & {\scriptsize 0.01s} & {\scriptsize 0.02s} & {\scriptsize 0.05s} & {\scriptsize 0.11s} & {\scriptsize 1.34s} \\
\hline
$\scriptstyle\mathsf{zls}^{n}(w)$ & $\scriptstyle\mathsf{wp}^n(\textbf{zdisp}, \emp \land w = \nil)$  & {\scriptsize 0.01s} & {\scriptsize 0.02s} & {\scriptsize 0.11s} & {\scriptsize 0.43s} & {\scriptsize 24.13s} \\
\hline
$\scriptstyle\mathsf{ls}^{n}(u) * \mathsf{ls}^{0}(v)$ & $\scriptstyle\mathsf{wp}(\mathbf{rev}, \mathsf{ls}^{n-1}(u) * \mathsf{ls}^{1}(v))$ & {\scriptsize 0.06s} & {\scriptsize 0.08s} & {\scriptsize 0.14s} & {\scriptsize 0.30s} & {\scriptsize 2.83s}  \\
\hline
$\scriptstyle\mathsf{ls}^{n}(u) * \mathsf{ls}^{0}(v)$ & $\scriptstyle\mathsf{wp}^n(\mathbf{rev}, u = \nil \land \mathsf{ls}^{n}(v))$ & {\scriptsize 0.06s} & {\scriptsize 0.12s} & {\scriptsize 0.56s} & {\scriptsize 1.75s} & {\scriptsize 27.82s}  \\
\hline
$\scriptstyle\mathsf{zls}^{n}(u) * \mathsf{zls}^{0}(v)$ & $\scriptstyle\mathsf{wp}(\mathbf{zrev}, \mathsf{zls}^{n-1}(u) * \mathsf{zls}^{1}(v))$ & {\scriptsize 0.22s} & {\scriptsize 0.04s} & {\scriptsize 0.12s} & {\scriptsize 0.25s} & {\scriptsize 2.16s}  \\
\hline
$\scriptstyle\mathsf{zls}^{n}(u) * \mathsf{zls}^{0}(v)$ & $\scriptstyle\mathsf{wp}^n(\mathbf{zrev}, u = \nil \land \mathsf{zls}^{n}(v))$ & {\scriptsize 0.04s} & {\scriptsize 0.10s} & {\scriptsize 0.41s} & {\scriptsize 1.27s} & {\scriptsize 20.26s}  \\
\hline
\end{tabular}\\
\caption{Experimental results}\label{tab:experiments}
\vspace*{-\baselineskip}
\end{table}

The first set of experiments is reported in Table
\ref{tab:experiments}. We have considered the following inductive
predicates commonly used as verification benchmarks \cite{sl-comp14}.
Here we check the validity of the entailment between $\mathsf{lhs}$
and $\mathsf{rhs}$, where both predicates are unfolded $n=1,2,3,4,8$
times. The second set of experiments, reported in Table
\ref{tab:experiments}, considers the verification conditions of the
forms $\varphi \Rightarrow \mathsf{wp}(\mathbf{l},\phi)$ and $\varphi
\Rightarrow \mathsf{wp}^n(\mathbf{l},\phi)$, where
$\mathsf{wp}(\mathbf{l},\phi)$ denotes the weakest precondition of the
$\seplog$ formula $\phi$ with respect to the sequence of statements
$\mathbf{l}$, and $\mathsf{wp}^n(\mathbf{l},\phi) =
\mathsf{wp}(\mathbf{l},
\ldots\mathsf{wp}(\mathbf{l},\mathsf{wp}(\mathbf{l},\phi)) \ldots)$
denotes the iterative application of the weakest precondition $n$
times in a row. We consider the loops depicted in Figure
\ref{fig:loops}, where, for each loop {\bf l} we consider the variant
    {\bf zl} as well, which tests that the data values contained
    within the memory cells are $0$, by the assertions on line 2. The
    postconditions are specified by finite unfoldings of the inductive
    predicates $\mathsf{ls}$ and $\mathsf{zls}$ (Figure
    \ref{fig:loops}).

\section{Conclusions}

We have presented a decision procedure for quantifier-free
$\seplogl{T}$ formulas that relies on a efficient,
counterexample-guided approach for establishing the $T$-satisfiability
of formulas having quantification over bounded sets.  We have
described an implementation of the approach as an integrated subsolver
in the DPLL($T$)-based SMT solver CVC4, showing the potential of the
procedure as a backend for tools reasoning about low-level pointer and
data manipulations.  For future work, we would like to extend the
approach to be applicable to cases where $\locs$ may have a finite
interpretation, and improve the heuristics used by our decision
procedure for choosing quantifier instantiations, in particular for
cases with many quantifier alternations.


\renewcommand{\baselinestretch}{0.90}
\bibliographystyle{splncs03}
\bibliography{refs}

\begin{thebibliography}{10}
\providecommand{\url}[1]{\texttt{#1}}
\providecommand{\urlprefix}{URL }

\bibitem{BansalPhd}
Bansal, K.: Decision Procedures for Finite Sets with Cardinality and Local
  Theory Extensions. {Ph.D.}\ thesis, New York University (2016)

\bibitem{bansal2015deciding}
Bansal, K., Reynolds, A., King, T., Barrett, C., Wies, T.: Deciding local
  theory extensions via e-matching. In: Computer Aided Verification, pp.
  87--105. Springer International Publishing (2015)

\bibitem{CVC4-CAV-11}
Barrett, C., Conway, C., Deters, M., Hadarean, L., Jovanovic, D., King, T.,
  Reynolds, A., Tinelli, C.: {CVC4}. In: Computer Aided Verification (CAV).
  Springer (2011)

\bibitem{smtlib25}
Barrett, C., Fontaine, P., Tinelli, C.: The {SMT-LIB} standard---{V}ersion 2.5.
  Tech. rep., The University of Iowa (2015), available at
  \url{http://smt-lib.org/}

\bibitem{SpaceInvader}
Berdine, J., Calcagno, C., Cook, B., Distefano, D., O'Hearn, P., Wies, T.,
  Yang, H.: Shape analysis for composite data structures. In: Proc. CAV'07.
  LNCS, vol. 4590. Springer (2007)

\bibitem{bjornerplaying}
Bj{\o}rner, N., Janota, M.: Playing with quantified satisfaction

\bibitem{cyclist}
Brotherston, J., Gorogiannis, N., Petersen, R.L.: A generic cyclic theorem
  prover. In: APLAS. pp. 350--367 (2012)

\bibitem{Infer}
Calcagno, C., Distefano, D.: Infer: An automatic program verifier for memory
  safety of c programs. In: Proc. of NASA Formal Methods'11. LNCS, vol. 6617.
  Springer (2011)

\bibitem{CalcagnoYangOHearn01}
Calcagno, C., Yang, H., O’hearn, P.W.: Computability and complexity results
  for a spatial assertion language for data structures. In: FST TCS 2001:
  Foundations of Software Technology and Theoretical Computer Science, pp.
  108--119. Springer Berlin Heidelberg (2001)

\bibitem{spen}
Enea, C., Sighireanu, M., Wu, Z.: On automated lemma generation for separation
  logic with inductive definitions. In: Automated Technology for Verification
  and Analysis - 13th International Symposium, {ATVA} 2015, Shanghai, China,
  October 12-15, 2015, Proceedings. pp. 80--96 (2015)

\bibitem{MeryGalmiche07}
Galmiche, D., M{\'e}ry, D.: Tableaux and resource graphs for separation logic.
  Journal of Logic and Computation  20(1),  189--231 (2010)

\bibitem{GanzingerHagenNieuwenhuisOliverasTinelli04}
Ganzinger, H., Hagen, G., Nieuwenhuis, R., Oliveras, A., Tinelli, C.: Dpll (t):
  Fast decision procedures. In: Computer aided verification, pp. 175--188.
  Springer (2004)

\bibitem{GeDeM-CAV-09}
Ge, Y., de~Moura, L.: Complete instantiation for quantified formulas in
  satisfiability modulo theories. In: Proceedings of CAV'09. LNCS, vol. 5643.
  Springer (2009)

\bibitem{slide}
Iosif, R., Rogalewicz, A., Vojnar, T.: Slide: Separation logic with inductive
  definitions, \\URL: {\ttfamily
  http://www.fit.vutbr.cz/research/groups/verifit/tools/slide/}

\bibitem{IshtiaqOHearn01}
Ishtiaq, S.S., O'Hearn, P.W.: Bi as an assertion language for mutable data
  structures. In: ACM SIGPLAN Notices. vol.~36, pp. 14--26. ACM (2001)

\bibitem{NavarroPerez2011}
Navarro~P{\'e}rez, J.A., Rybalchenko, A.: Separation logic + superposition
  calculus = heap theorem prover. SIGPLAN Not.  46(6),  556--566 (Jun 2011),
  \url{http://doi.acm.org/10.1145/1993316.1993563}

\bibitem{Sleek}
Nguyen, H.H., Chin, W.N.: Enhancing program verification with lemmas. In: Proc
  of CAV'08. LNCS, vol. 5123. Springer (2008)

\bibitem{Piskac2013}
Piskac, R., Wies, T., Zufferey, D.: Computer Aided Verification: 25th
  International Conference, CAV 2013, Saint Petersburg, Russia, July 13-19,
  2013. Proceedings, chap. Automating Separation Logic Using SMT, pp. 773--789.
  Springer Berlin Heidelberg, Berlin, Heidelberg (2013)

\bibitem{ReynoldsDKBT15Cav}
Reynolds, A., Deters, M., Kuncak, V., Barrett, C.W., Tinelli, C.:
  Counterexample guided quantifier instantiation for synthesis in {CVC4}. In:
  CAV. Springer (2015)

\bibitem{RKK2015}
Reynolds, A., King, T., Kuncak, V.: An instantiation-based approach for solving
  quantified linear arithmetic. CoRR  abs/1510.02642 (2015)

\bibitem{Reynolds02}
Reynolds, J.: {Separation Logic: A Logic for Shared Mutable Data Structures}.
  In: Proc. of LICS'02. IEEE CS Press (2002)

\bibitem{suter2011sets}
Suter, P., Steiger, R., Kuncak, V.: Sets with cardinality constraints in
  satisfiability modulo theories. In: Verification, Model Checking, and
  Abstract Interpretation, pp. 403--418. Springer Berlin Heidelberg (2011)

\bibitem{YangPhD}
Yang, H.: Local Reasoning for Stateful Programs. {Ph.D.}\ thesis, University of
  Illinois at Urbana-Champaign (2001)

\end{thebibliography}


\appendix 
\section{Proofs}

\subsection{Proof of Lemma \ref{lemma:labeling}}

\proof{ By induction on the structure of $\varphi$. We distinguish the
  following cases: \begin{compactitem}[-]
  \item $\varphi \equiv \phi * \psi$ and $\I,h \models_\seplog \phi *
    \psi$ iff there exist heaps $g_1, g_2$ such that $h=g_1 \uplus
    g_2$ and $\I,g_1 \models_\seplog \phi$, $\I,g_2 \models_\seplog
    \psi$. ``$\Rightarrow$'' Let $\tup{h}$ and $\tup{h'}$ be tuples of
    heaps satisfying conditions (\ref{it1:labeling}) and
    (\ref{it2:labeling}). By the induction hypothesis we obtain:
    \[\begin{array}{rcl}
    \I[\lterm_1\leftarrow\dom(g_1)][\tup\lterm\leftarrow\dom(\tup{h})][\tup\pto \leftarrow \tup{h'}] & \models_T & \labd{\phi \lab [\lterm_1\cap\tup\lterm,\tup\pto]} \\
    \I[\lterm_2\leftarrow\dom(g_2)][\tup\lterm\leftarrow\dom(\tup{h})][\tup\pto \leftarrow \tup{h'}] & \models_T & \labd{\phi \lab [\lterm_2\cap\tup\lterm,\tup\pto]} 
    \end{array}\]
    because $g_j \cap h_i \subseteq h'_i$ for each $j=1,2$ and
    $i=1,\ldots,n$. Since, moreover, $\dom(g_1) \cap \dom(g_2) =
    \emptyset$ and $\dom(g_1) \cup \dom(g_2) = \bigcup_{i=1}^n
    \dom(h_i)$, we obtain
    \(\I[\tup\lterm\leftarrow\dom(\tup{h})][\tup\pto\leftarrow\tup{h'}]
    \models_T \labd{\phi*\psi \lab [\lterm,\pto]}\). ``$\Leftarrow$'' If
    \(\I[\tup\lterm\leftarrow\dom(\tup{h})][\tup\pto\leftarrow\tup{h'}]
    \models_T \labd{\phi*\psi \lab [\lterm,\pto]}\), there exists sets
    $L_1, L_2 \subseteq \locs$ such that $L_1 \cap L_2 = \emptyset$
    and $\dom(h)=L_1 \cup L_2$. Let $g_1 = \proj{h}{L_1}$ and $g_2 =
    \proj{h}{L_2}$. We have that $h = g_1 \uplus g_2$ and:
    \[\begin{array}{rcl}
    \I[\lterm_1\leftarrow\dom(g_1)][\tup\lterm\leftarrow\dom(\tup{h})][\tup\pto\leftarrow\tup{h'}] & \models_T & \labd{\phi \lab [\lterm_1\cap\tup\lterm,\tup\pto]} \\
    \I[\lterm_2\leftarrow\dom(g_2)][\tup\lterm\leftarrow\dom(\tup{h})][\tup\pto\leftarrow\tup{h'}] & \models_T & \labd{\psi \lab [\lterm_2\cap\tup\lterm,\tup\pto]}
    \end{array}\]
    Since, moreover, $g_j = \biguplus_{i=1}^n g_j \cap h_i$ and $g_j
    \cap h_i \subseteq h'_i$, for $j=1,2$ and $i=1,\ldots,n$ we can
    apply the induction hypothesis to obtain that $\I,g_1
    \models_\seplog \phi$ and $\I,g_2 \models_\seplog \psi$, thus
    $\I,h \models_\seplog \phi*\psi$.
    \item $\varphi \equiv \phi\wand\psi$. ``$\Rightarrow$'' Suppose
      that $\I,h \models_\seplog \phi\wand\psi$ and let $g \subseteq
      g'$ be heaps such that $g \# h$ and
      $\I[\lterm'\leftarrow\dom(g)][\pto'\leftarrow g'] \models_T
      \labd{\phi \lab [\lterm',\pto']}$. By the induction hypothesis,
      we obtain $\I,g \models_\seplog \phi$, thus $\I,g\uplus h
      \models_\seplog \psi$. Since, moreover, $g\cdot\tup{h}$ and
      $g'\cdot\tup{h'}$ satisfy the conditions (\ref{it1:labeling})
      and (\ref{it2:labeling}), by the induction hypothesis, we obtain
      \(\I[\lterm'\leftarrow\dom(g)][\pto'\leftarrow
        g'][\tup\lterm\leftarrow\dom(\tup{h})][\pto\leftarrow\tup{h}']
      \models_T \labd{\psi \lab
        [\lterm'\cdot\tup\lterm,\pto'\cdot\tup{\pto}]}\). Since the
      choice of $g$ and $g'$ was arbitrary, we obtain that
      \(\I[\tup\lterm\leftarrow\dom(\tup{h})][\pto\leftarrow\tup{h}']
      \models_T \labd{\phi \wand \psi \lab
        [\tup\lterm,\tup{\pto}]}\). ``$\Leftarrow$'' Suppose that
      \(\I[\tup\lterm\leftarrow\dom(\tup{h})][\pto\leftarrow\tup{h}']
      \models_T \labd{\phi \wand \psi \lab [\tup\lterm,\tup{\pto}]}\)
      and let $g \subseteq g'$ be heaps such that $g \# h$ and $\I,g
      \models_\seplog \phi$. By the induction hypothesis, we have that
      \(\I[\lterm'\leftarrow\dom(g)][\pto'\leftarrow g'] \models_T
      \labd{\phi \lab [\lterm',\pto']}\) and since $\dom(g) \cap
      (\bigcup_{i=1}^n\dom(h_i)) = \emptyset$, we have
      \(\I[\lterm'\leftarrow\dom(g)][\pto'\leftarrow
        g'][\tup\lterm\leftarrow\tup{h}][\tup\pto\leftarrow\tup{h}']
      \models_T \labd{\psi \lab
        [\lterm'\cdot\tup\lterm,\pto'\cdot\tup\pto]}\). By the
      induction hypothesis, we obtain $\I,g\uplus h \models_\seplog
      \psi$, thus $\I,h \models_\seplog \phi\wand\psi$. 
    \item $\varphi \equiv \tterm \mapsto \uterm$. We have $\I,h
      \models_\seplog \tterm \mapsto \uterm$ iff
      $h=\{(\tterm^\I,\uterm^I)\}$ and
      $\tterm^\I\not\teq\nil$. ``$\Rightarrow$'' Let $\tup{h}$ and
      $\tup{h}'$ be tuples of heaps satisfying the conditions
      (\ref{it1:labeling}) and (\ref{it2:labeling}). Then there exists
      $1 \leq i \leq n$ such that $\dom(h_i) = \{\tterm^I\}$,
      $h'_i(\tterm^I)=\uterm^I$ and $\dom(h_j)=\emptyset$ for all $j
      \in\set{1,\ldots,n} \setminus \set{i}$. We obtain, consequently
      that
      \(\I[\tup\lterm\leftarrow\dom(\tup{h})][\pto\leftarrow\tup{h}']
      \models_T \bigcup \tup\lterm=\set{\tterm} \wedge
      \tupite{\tterm\in\tup\lterm}{\tup\pto(\tterm)=\uterm} \wedge
      \tterm\not\teq\nil\). ``$\Leftarrow$'' If
      $\I[\tup\lterm\leftarrow\dom(\tup{h})][\tup\pto\leftarrow\tup{h}']
      \models_T \labd{\tterm \mapsto \uterm \lab [\tup\lterm,\tup\pto]}$ for
      each $\tup{h}$ and $\tup{h}'$ satisfying conditions
      (\ref{it1:labeling}) and (\ref{it2:labeling}), we easily obtain
      that $\{\tterm^\I\} = \dom(h_i) \subseteq \dom(h'_i)$ and
      $h'_i(\tterm^\I)=\uterm^\I$ for some $1 \leq i \leq n$, leading
      to $h=\{(\tterm^\I,\uterm^\I)\}$. Moreover
      $\tterm^\I\not\teq\nil$, thus $\I,h \models_\seplog \tterm
      \mapsto \uterm$.
    \item $\varphi \equiv \emp$. We have $\I,h \models_\seplog \emp$
      iff $\dom(h) = \emptyset$. ``$\Rightarrow$'' For any tuples
      $\tup{h}$ satisfying condition (\ref{it1:labeling}) we have
      $\dom(h_i) = \emptyset$ for all $1 \leq i \leq n$, thus
      $\I[\tup\lterm\leftarrow\dom(\tup{h})][\tup\pto\leftarrow\tup{h}']
      \models_T \bigcup\tup\lterm=\emptyset$ for all $\tup{h}'$
      satisfying condition (\ref{it2:labeling}). ``$\Leftarrow$ Let
      $\tup{h}$ and $\tup{h}'$ be tuples of heaps satisfying
      conditions (\ref{it1:labeling}) and (\ref{it2:labeling}), such
      that $\I[\tup\lterm\leftarrow\dom(\tup{h})][\tup\pto\leftarrow
        \tup{h}'] \models_T \labd{\emp \lab [\tup\lterm,\tup\pto]}$,
      then $\dom(h)=\emptyset$ and $\I,h \models_\seplog \emp$.    
  \end{compactitem}
  The cases $\varphi \equiv \phi \wedge \psi$, $\varphi \equiv
  \neg\phi$ and $\varphi$ is a $\Sigma$-formula are an easy
  exercise. \qed}

\subsection{Proof of Lemma \ref{lemma:bounded-labeling}}

Its proof relies on the following technical fact \cite{YangPhD}:

\begin{proposition}\label{prop:bounded-wand}
  Given two formulae $\phi,\psi \in \seplog(T)$, for any
  interpretation $\I$ of $T$, let \(L \subseteq \locs^\I \setminus
  \terms{\phi\wand\psi}^\I\) be a set such that
  $\card{L}=\max(\len{\phi},\len{\psi})$ and $v \in \data^\I \setminus
  \terms{\phi\wand\psi}^\I$ be a data value. Then, for any heap $h$ we
  have $\I,h \models_\seplog \phi \wand \psi$ if and only if, for any
  heap $h'$ such that $h \# h'$, $\I,h' \models_\seplog \phi$,
  and: \begin{compactenum}
  \item $\dom(h') \subseteq L \cup \terms{\phi\wand\psi}^\I$,
  \item $h'(\ell) \in \terms{\phi\wand\psi}^\I \cup
    \set{v}$, for all $\ell \in \dom(h')$,
  \end{compactenum}
  we have that $\I,h\uplus h' \models_\seplog \psi$.
\end{proposition}
\proof{See \cite[Proposition 89]{YangPhD}. \qed}

\proof{By Lemma \ref{lemma:decidability}, there exists $h$ such that
  $\I,h \models_\seplog \varphi$ if and only if there exists $h'$ such
  that $\dom(h') \subseteq L \cup \terms{\varphi}^\I$ and for all
  $\ell\in\dom(h')$, $h'(\ell) \in \terms{\varphi}^\I \cup
  \set{v}$, for a value $v \in \data^\I \setminus
  \terms{\varphi}^\I$. It is thus sufficient to
  prove the statement for these heaps only and assume from now on that
  $h$ satisfies the conditions of Lemma \ref{lemma:decidability}. We
  prove the following stronger statement: 

  \begin{fact}
    For all heaps $h$ such that: \begin{compactitem}[-]
    \item $\dom(h) \subseteq L \cup \terms{\varphi}^\I$ and for all
      $\ell\in\dom(h)$, $h(\ell) \in \terms{\varphi}^\I \cup \set{v}$,
    \end{compactitem}
    we have $\I,h \models_\seplog \varphi$ iff for all tuples $\tup{h} =
    \tuple{h_1,\ldots,h_n}$ and $\tup{h'} = \tuple{h'_1,\ldots,h'_n}$
    such that: \begin{compactitem}[-]
    \item $h=h_1 \uplus \ldots \uplus h_n$ and $h_1 \subseteq h'_1,
      \ldots, h_n \subseteq h'_n$,
    \end{compactitem} 
    we have \(\I[\tup\lterm\leftarrow \dom(\tup{h})][\tup\pto\leftarrow
      \tup{h}'][C \leftarrow L][d \leftarrow v] \models_T \labd{\varphi \lab
      [\tup\lterm,\tup\pto,C]}\enspace\).
  \end{fact}

  The proof is by induction on the structure of $\varphi$, among the
  lines of the proof of Lemma \ref{lemma:labeling}. In particular, it
  can be easily verified that the quantified sets variables and
  uninterpreted functions belong to the domains required by the rules
  of Figure \ref{fig:bounded-labeling}. \qed}

\subsection{Proof of Theorem \ref{thm:solveslt}}

\begin{lemma}
\label{lem:purify} For all $T$-formulae $\varphi$:
\begin{enumerate}
\item \label{it:purify-ent} $\varphi$ is $T$-satisfiable only if $\purifyrec{ \varphi }$ is $T$-satisfiable, and
\item \label{it:purify-eq} $\varphi$ and $\unpurify{ \purifyrec{ \varphi } }$ are $T$-equivalent up to their shared variables.
\end{enumerate}
\end{lemma}
\begin{proof} (Sketch)
To show Part (\ref{it:purify-ent}),
let $\I$ be a model of $T$ and $\varphi$.
Let $\J$ be an extension of $\I$ such that $A^\J = (\forall \vec x. \psi)^\I$
for each $\forall \vec x. \psi \in \quants{ \purifyrec{ \varphi } }$ where $( A, \vec k ) \leftrightharpoons \forall \vec x. \psi$.
The interpretation $\J$ satisfies $T$ and $\purifyrec{ \varphi }$
To show Part (\ref{it:purify-eq}), notice that
$\unpurify{ \purifyrec{ \varphi } }$ is a set that can be constructed from an initial value $\{ \varphi \}$, 
and updated by adding formulas of the form $( \neg \forall \vec x. \psi \Rightarrow \neg \psi[ \vec k/\vec x]  )$ for fresh constants $\vec k$.
For each step of this construction, it can be shown that the set of models are same when restricted to the interpretation of all variables apart from $\vec k$.
Thus, by induction, $\varphi$ and $\unpurify{ \purifyrec{ \varphi } }$ are $T$-equivalent up to their shared variables.
\qed
\end{proof}

\begin{lemma}
\label{lem:gamma-minimal0} 
For every recursive call to $\solveslt{T}( \Gamma )$,
if $\Gamma, A \models_T \purify{ \psi[ \vec k/\vec x] }$
where $( A, \vec k ) \leftrightharpoons \forall \vec x. \psi$,
then $\Gamma \models_T \forall \vec x. \psi$.
\end{lemma}
\begin{proof}
It suffices to show there exists a subset $\Gamma'$ of $\Gamma$ such that
$\Gamma'$ does not contain $\vec k$ and
$\Gamma' \models_T \purify{ \psi[ \vec k/\vec x] }$
(if such a $\Gamma'$ exists, then since $\Gamma' \subseteq \Gamma$ and $\Gamma'$ does not contain $\vec k$,
we have that $\Gamma \models_T \forall \vec x. \psi$).
By construction, $\Gamma$ may be partitioned into sets $\Gamma'$ and $\Gamma''$,
where $\Gamma'$ does not contain $\vec k$, and $\Gamma''$ contains only:
\begin{enumerate}
\item $\neg A \Rightarrow \purify{ \neg \psi[ \vec k/\vec x] }$, and
\item Constraints of the form $\neg A_1 \Rightarrow \purify{ \neg \psi_1[ \vec j/\vec y] }$ and $\purifyrec{ A_1 \Rightarrow \psi_1[ \vec t/\vec y] }$, 
where $( A_1, \vec j ) \leftrightharpoons \forall \vec y. \psi_1$ and
$A_1$ does not occur in $\Gamma'$.
\end{enumerate}
Assume that $\Gamma \setminus \Gamma'', A, \purify{ \neg \psi[ \vec k/\vec x] }$ has a model $\I$.
Let $\J$ be an extension of $\I$ such that 
for each $( A_1, \vec j ) \leftrightharpoons \forall \vec x. \psi_1$ occurring in $\quants{ \Gamma'' }$ but not in $\quants{ \Gamma' }$,
we have $A_1^\J = ( \forall \vec x. \psi_1 )^\I$.
The interpretation $\J$ satisfies $\Gamma, A, \purify{ \neg \psi[ \vec k/\vec x] }$,
noting that the constraint $\neg A \Rightarrow \purify{ \neg \psi[ \vec k/\vec x] }$ holds since $A^\J$ must be $\top$.
This contradicts the assumption that  $\Gamma, A \models_T \purify{ \psi[ \vec k/\vec x] }$,
and thus $\Gamma \setminus \Gamma'', A, \purify{ \neg \psi[ \vec k/\vec x] }$ is $T$-unsatisfiable,
in other words, $\Gamma', A \models_T \purify{ \psi[ \vec k/\vec x] }$.
Since $A$ does not occur in $\purify{ \psi[ \vec k/\vec x] }$,
we have that $\Gamma' \models_T \purify{ \psi[ \vec k/\vec x] }$, and hence the lemma holds.
\qed
\end{proof}

\begin{lemma}
\label{lem:gamma-minimal} 
If $\Gamma = \{ \purifyrec{ \varphi } \mid \varphi \in S \}$ for some set $S$,
and $\forall \vec x. \psi \in \quants{\Gamma}$ is $\ord{\Gamma}{\I}$-minimal,
then $\Gamma \cup \{ (\neg)\purify{ \psi[ \vec k/\vec x] } \}$ and
$\Gamma \cup \{ (\neg)\psi[ \vec k/\vec x] \}$ are $T$-equivalent up to their shared variables.
\end{lemma}
\begin{proof}
(Sketch)
By definition of $\ord{\Gamma}{\I}$-minimal,
we have that $\Gamma, A_0 \models_T \purify{ \psi_0[ \vec k/\vec x] }$
for all $\forall \vec x_0. \psi_0 \in Q^+_\I( \purify{ \psi[ \vec k_0/\vec x_0] } )$
where $( A_0, \vec k_0 ) \leftrightharpoons \forall \vec x_0. \psi_0$.
For each such formula, by Lemma~\ref{lem:gamma-minimal0},
$\Gamma \models_T \forall \vec x_0. \psi_0$.
Thus, 
$\Gamma \cup \{ (\neg) \purify{ \psi[ \vec k/\vec x] } \}$ and 
$\Gamma \cup \{ (\neg) \unpurify{ \purify{ \psi[ \vec k/\vec x] } } \}$
are $T$-equivalent up to their shared variables,
which by Lemma~\ref{lem:purify}.\ref{it:purify-eq} implies that
$\Gamma \cup \{ (\neg)\purify{ \psi[ \vec k/\vec x] } \}$ and
$\Gamma \cup \{ (\neg)\psi[ \vec k/\vec x] \}$ are $T$-equivalent up to their shared variables.
\qed
\end{proof}

\begin{lemma}
\label{lem:solvet}
For all $\seplogl{T}$ formulae $\varphi$,
$\solveslt{T}( \purifyrec{ \psi } )$ where $\psi$ is $\labd{(\varphi \lab [\lterm,\pto,C])}$:
\begin{enumerate}
\item \label{it:tunsat} Answers ``unsat" only if $\psi$ is $T$-unsatisfiable.
\item \label{it:tsat} Answers ``sat" only if $\psi$ is $T$-satisfiable.
\item \label{it:tterm} Terminates.
\end{enumerate}
\end{lemma}
\begin{proof}
Assume $\solveslt{T}( \Gamma_i )$ calls $\solveslt{T}( \Gamma_{i+1} )$ for $i=0, \ldots, n-1$,
where $n$ is finite and $\Gamma_0 = \purifyrec{ \psi }$.
By definition of $\solveslt{T}$ and Lemma~\ref{lem:purify} (\ref{it:purify-eq}),
it can be shown that 
$\unpurify{ \Gamma_i }$ and $\unpurify{ \Gamma_{i+1} }$ are equisatisfiable in $T$,
and thus by induction
$\unpurify{ \Gamma_j }$ and $\unpurify{ \Gamma_k }$ are equisatisfiable in $T$ for each $j, k \in \{ 1, \ldots, n \}$.

To show Part (\ref{it:tunsat}), assume without loss of generality, that
$\solveslt{T}( \Gamma_n )$ answers ``unsat".
Then $\Gamma_n$ is $T$-unsatisfiable, and
by Lemma~\ref{lem:purify}.\ref{it:purify-ent}, we have that $\unpurify{ \Gamma_n }$ is $T$-unsatisfiable.
Thus, $\unpurify{ \Gamma_0 }$ is $T$-unsatisfiable, and thus by 
Lemma~\ref{lem:purify}.\ref{it:purify-eq}, $\psi$ is $T$-unsatisfiable.

To show Part (\ref{it:tsat}), without loss of generality,
$\solveslt{T}( \Gamma_n )$ answers ``sat".
Then $\Gamma_n$ is $T$-satisfiable with model $\I$.
We argue that $\I$ satisfies $\unpurify{\Gamma_n}$,
which implies that $\unpurify{ \Gamma_0 }$ is $T$-satisfiable, and thus by 
Lemma~\ref{lem:purify}.\ref{it:purify-eq}, $\psi$ is $T$-satisfiable.
Assume that $\I$ satisfies $\Gamma_n$ but does not satisfy $\unpurify{\Gamma_n}$.
Thus,
$A^\I \neq (\forall \vec \qvar. \psi)^\I$
for some $\forall \vec \qvar. \psi \in \quants{\Gamma_n}$ where $(A, \vec k) \leftrightharpoons \forall \vec \qvar. \psi$.
In the case that $A^\I = \bot$ and $(\forall \vec \qvar. \psi)^\I = \top$,
note that $\I$ does not satisfy the formula $( \neg A \Rightarrow \purify{ \neg \varphi[ \vec k/\vec \qvar] } ) \in \Gamma_n$.
In the case that $A^\I = \top$ and $(\forall \vec \qvar. \psi)^\I = \bot$,
by the definition of $\solveslt{T}$,
we have $\Gamma_n, A \models_T \purify{ \psi[ \vec k/\vec \qvar] }$.
By Lemma~\ref{lem:gamma-minimal0},
we have that $\Gamma_n \models_T \forall \vec x. \psi$.
Since $\I$ is a model of $\Gamma_n$, it must be the case that $(\forall \vec x. \psi)^\I = \top$, contradicting the assumption that
$(\forall \vec \qvar. \psi)^\I = \bot$.
This contradicts the assumption that $\I$ does not satisfy $\unpurify{\Gamma_n}$.

To show Part (\ref{it:tterm}),
first note that the checks for $T$-satisfiability and $T$-entailment terminate,
by assumption of a decision procedure for the $T$-satisfiability of quantifier-free formulae.
Moreover, for each quantified formula $\forall \vec \qvar. \psi$, 
only a finite number of instances $A \Rightarrow \psi[ \vec t/ \vec \qvar ]$ exist for which the algorithm will
add $\purifyrec{ A \Rightarrow \psi[ \vec t/ \vec \qvar ] }$ to $\Gamma$,
which implies that the algorithm will consider only a finite number of quantified formulae,
each for which only a finite number of instances will be added in this way.
Thus, it suffices to show that the algorithm adds to $\Gamma$ only sets
$\purifyrec{ A \Rightarrow \psi[ \vec t/ \vec \qvar ] }$ that are not a subset of $\Gamma$.
Assume this is not the case for some 
$\purifyrec{ A \Rightarrow \psi[ \vec t/ \vec \qvar ] }$,
where $\forall \vec \qvar. \psi \in \quants{\Gamma_n}$, $\forall \vec \qvar. \psi$ is $\ord{\Gamma}{\I}$-minimal, 
and $(A, \vec k) \leftrightharpoons \forall \vec \qvar. \psi$.
The terms $\vec t$ meet the criteria in the procedure, in particular,
for some model $\J$ of $\Gamma \cup \{ \neg \purify{\psi[ \vec k/\vec \qvar]} \}$ where $A^\J = \top$,
we have $\vec t^\J = \vec k^\J$.
Since $A^\J = \top$ and $(A \Rightarrow \purify{ \psi[ \vec t/ \vec \qvar ] }) \in \purifyrec{ A \Rightarrow \psi[ \vec t/ \vec \qvar ] } \subseteq \Gamma$,
$\J$ satisfies $\purify{\psi[ \vec t/ \vec \qvar ] }$.
Also, $\J$ satisfies $\neg \purify{ \psi[ \vec k/\vec \qvar] }$.
By Lemma~\ref{lem:gamma-minimal},
since $\forall \vec \qvar. \psi$ is $\ord{\Gamma}{\I}$-minimal,
$\J$ satisfies $\psi[ \vec t/ \vec \qvar ]$ and $\neg \psi[ \vec k/\vec \qvar]$.
Thus, $\vec k^\J \neq \vec t^\J$, contradicting our assumption.
\qed
\end{proof}

\begin{proof}
To show Part (\ref{it:sltunsat}), 
by Lemma~\ref{lem:solvet}(\ref{it:tunsat}),
it must be the case that $\labd{(\varphi \lab [\lterm,\pto,C])}$ is $T$-unsatisfiable.
Thus, there cannot be heaps meeting the requirements of $\tup{h'}$ and $\tup{h''}$ in Lemma~\ref{lemma:bounded-labeling},
and by that lemma means that $\varphi$ is $( \seplog, T )$-unsatisfiable.
To show Part (\ref{it:sltsat}),
by Lemma~\ref{lem:solvet}(\ref{it:tsat}),
it must be the case that $\labd{(\varphi \lab [\lterm,\pto,C])}$ is $T$-satisfied by a model of $T$, call it $\J$.
Let $h'$ be the heap with domain $\lterm^\J$ such that $h'( u ) = \pto^\J( u )$ for all $u \in \lterm^\J$,
and let $h''=\pto^\J$.
Since $\locs$ has infinite cardinality, and due to the structure of $\labd{(\varphi \lab [\lterm,\pto,C])}$, 
we may assume that $C^\J \subseteq \locs^\J \setminus \terms{\varphi}^\J$,
call this set $L$.
Let $v = d^\J$.
Let $\I$ be an interpretation such that $\I[\lterm \leftarrow \dom(h')][\pto \mapsfrom h''][C \leftarrow L][ d \leftarrow v ] = \J$.
Since by assumption $\J \models_T \labd{(\varphi \lab [\lterm,\pto,C])}$, then by Lemma~\ref{lemma:bounded-labeling}, there exists a heap $h$ such that $\I, h \models_{\seplog} \varphi$,
and thus $\varphi$ is $(\seplog,T)$-satisfiable.
Part (\ref{it:sltterm}) is an immediate consequence of Lemma~\ref{lem:solvet}(\ref{it:tterm}).
\qed
\end{proof}

\subsection{Proof of Corollary \ref{cor:slt-pspace}}

\proof{PSPACE-hardness is by the reduction from QSAT, which
  generalizes \cite[Definition 7]{CalcagnoYangOHearn01} to our
  case. Let $\phi \equiv \forall x_1 \exists y_1 \ldots \forall x_n
  \exists y_n ~.~ \psi$ be an instance of QSAT, where $\psi$ is a
  boolean combination of the variables $x_i$ and $y_i$ of sort
  $\bools$. We encode $\phi$ in $\seplog(T)$ using the translation
  function $\mathrm{Tr}(.)$, defined by induction on the structure of
  $\phi$:
  \[\begin{array}{rclcrcl}
  \mathrm{Tr}(x_i) & \equiv & (\tterm_i \mapsto d) * \top & \hspace*{0.5cm} & \mathrm{Tr}(y_i) & \equiv & (\uterm_i \mapsto d) * \top \\
  \mathrm{Tr}(\neg \psi) & \equiv & \neg \mathrm{Tr}(\psi) & \hspace*{0.5cm} & \mathrm{Tr}(\psi_1 \bullet \psi_2) & \equiv & \mathrm{Tr}(\psi_1) \bullet \mathrm{Tr}(\psi_2) \\
  \mathrm{Tr}(\exists y_i ~.~ \psi) & \equiv & ((\uterm_i \mapsto d) \vee \emp) * \mathrm{Tr}(\psi) & \hspace*{0.5cm} &
  \mathrm{Tr}(\forall x_i ~.~ \psi) & \equiv & \neg(((\tterm_i \mapsto d) \vee \emp) * \neg\mathrm{Tr}(\psi))
  \end{array}\]
  where $d$ is constant of sort $\data$ and $\bullet \in
  \set{\wedge,\vee}$. It is not hard to check that $\phi = \top$ if
  and only if there exists a $T$-interpretation $\I$ and a heap $h$
  such that $\I,h \models_\seplog \mathrm{Tr}(\phi)$. 

  To prove that $(\seplog,T)$-satisfiability is in PSPACE, we analyse
  the space complexity of the $\solvesl{T}$ algorithm. First, for any
  $\seplog(T)$-formula $\varphi$, let $\size{\varphi}$ be the size of
  the syntax tree of $\varphi$. Clearly
  $\card{\lfloor\varphi\rfloor^*} \leq \size{\varphi}$, and moreover,
  $\size{\psi} \leq \size{\varphi}$ for each $\psi \in
  \lfloor\varphi\rfloor^*$. Then a representation of the set
  $\lfloor\varphi\rfloor^*$ by simply enumerating its elements will
  take space at most $\size{\varphi}^2$. For a set of formulae
  $\Gamma$, let $\size{\Gamma}=\sum_{\varphi\in\Gamma}\size{\varphi}$
  denote the size of its enumerative representation. 

  Second, it is not difficult to see that $\len{\varphi} \leq
  \size{\varphi}$ and $\card{\terms{\varphi}} \leq \size{\varphi}$,
  for any $\seplog(T)$-formula $\varphi$. Then, for each subformula
  $\forall x \subseteq s ~.~ \psi$ of $\labd{\varphi \lab
    [\ell,\pto,C]}$, we have $\card{s} \leq \size{\varphi}^2$ -- in
  fact $\card{s} \leq \size{\varphi}$ if $x$ is of sort $\sets(\locs)$
  and $\card{s} \leq \size{\varphi}^2$ if $x$ is of sort $\locs
  \mapsto \data$. Then there are at most $\size{\varphi}^2$ recursive
  calls on line 3 of the $\solveslt{T}$ procedure, that corresponds to
  an instance of the subformula $\forall x \subseteq s ~.~ \psi$ of
  $\labd{\varphi \lab [\ell,\pto,C]}$. Since there are at most
  $\size{\varphi}$ such subformulae, there are at most
  $\size{\varphi}^3$ recursive calls to $\solveslt{T}$ with arguments
  $\Gamma_0, \ldots, \Gamma_n$, respectively. Moreover, we have
  $\Gamma_0 = \set{\labd{\varphi \lab [\ell,\pto,C]}}$, thus
  $\size{\Gamma_0} = \mathcal{O}(\size{\varphi})$ and for each
  $i=0,\ldots,n-1$,
  $\size{\Gamma_{i+1}}\leq\size{\Gamma_i}+\size{\varphi}^2$, because a
  set of formulae $\lfloor A \Rightarrow
  \psi[\vec{t}/\vec{x}]\rfloor^*$ of size at most $\size{\varphi}^2$ is
  added to $\Gamma_i$. Because $T$-satisfiability is in PSPACE, by the
  hypothesis, the checks at lines 1 and 2 can be done within space
  bounded by a polynomial in $\size{\varphi}$, thus the space needed
  by $\solvesl{T}(\varphi)$ is also bounded by a polynomial in
  $\size{\varphi}$. Hence the $(\seplog,T)$-satisfiability problem is
  in PSPACE. \qed}

\end{document}